\pdfoutput=1
\RequirePackage{ifpdf}
\ifpdf 
\documentclass[pdftex]{sigma}
\else
\documentclass{sigma}
\fi

\numberwithin{equation}{section}

\newtheorem{Theorem}{Theorem}[section]

\newtheorem{cor}[Theorem]{Corollary}

\newtheorem{lem}[Theorem]{Lemma}

{\theoremstyle{definition}
\newtheorem{de}[Theorem]{Definition}
\newtheorem{rem}[Theorem]{Remark}

}

\def\Zb{\mathbf{Z}}
\def\ab{\mathbf{a}}
\def\bb{\mathbf{b}}

\def\pb{\mathbf{p}}

\def\rb{\mathbf{r}}

\def\tb{\mathbf{t}}

\def\xb{\mathbf{x}}

\def\zb{\mathbf{z}}

\def\det{\mathrm {det}}

\def\Pf{\mathrm {Pf}}

\newcommand{\Pa}{\mathrm{P}}
\newcommand{\DP}{\mathrm{DP}}
\newcommand{\OP}{\mathrm{OP}}

\def\bc{\begin{corollary}}
\def\ec{\end{corollary}}
\def\be{\begin{equation}}
\def\ee{\end{equation}}

\def\bea{\begin{eqnarray}}
\def\eea{\end{eqnarray}}

\def\PP{{\mathcal P}}
\def\FF {{\mathcal F}}

\DeclareMathOperator{\GL}{GL}
\DeclareMathOperator{\SO}{SO}
\def\gl{\mathfrak{gl}}
\DeclareMathOperator{\sgn}{sgn}

\begin{document}
\allowdisplaybreaks

\newcommand{\arXivNumber}{2401.06032}

\renewcommand{\PaperNumber}{098}

\FirstPageHeading

\ShortArticleName{Bilinear Expansions of KP Multipair Correlators in BKP Correlators}

\ArticleName{Bilinear Expansions of KP Multipair Correlators\\ in BKP Correlators}

\Author{Aleksandr Yu.~ORLOV~$^{\rm ab}$}

\AuthorNameForHeading{A.Yu.~Orlov}

\Address{$^{\rm a)}$~National Research University Higher School of Economics, Moscow, Russia}
\EmailD{\href{mailto:orlovs55@mail.ru}{orlovs55@mail.ru}}

\Address{$^{\rm b)}$~Institute of Oceanology, Moscow, Russia}

\ArticleDates{Received February 03, 2024, in final form October 08, 2024; Published online October 31, 2024}

\Abstract{I present a generalization of our joint works with John Harnad (2021) that relates Schur functions, KP tau functions and KP correlation functions to Schur's $Q$-functions, BKP tau functions and BKP correlation functions, respectively.}


\Classification{05E05; 37K10; 81R12}

\begin{flushright}
\begin{minipage}{60mm}
\it Dedicated to John Harnad\\ on the occasion of his 75th birthday
\end{minipage}
\end{flushright}

\subsection*{Extended abstract}

The relationship between the Schur function $s_{\lambda}(\mathbf{z})$ and the projective Schur function $Q_\alpha(\mathbf{x})$ (the same Schur's $Q$-functions) is well known in the case where the partitions $\lambda$ and $\alpha$ are related as follows: the partition $\lambda$ has a special type, namely, it is the double $D(\alpha)$ of the strict partition $\alpha$, and at the same time the argument $\zb$ of the Schur function is also special: it can be written as a supersymmetric Newton sum $D(\mathbf{x})=\mathbf{x}/{-}\mathbf{x}$ (one could say that the argument of the Schur function must be the ``double'' of the argument of the projective Schur function). This relationship looks like this: $s_{D(\alpha)} (D(\mathbf{x}) )= (Q_\alpha(\mathbf{x}) )^2$, see \cite{Mac1}. In \cite{HO1}, we obtained a more general bilinear relationship between $s_\lambda$ and $Q_\alpha$ by removing the mentioned restriction on the partition of $\lambda$. In a similar way, bilinear relations were obtained between the determinant and Pfaffian tau functions (namely, between the tau functions of the KP and BKP hierarchies). In this paper, we remove restrictions from the argument of the Schur function, which no longer has to be the “double” of the argument of the projective Schur function, and obtain the most general connection between $s_\lambda$ and $Q_\alpha$, as well as between the determinant and Pfaffian tau functions and correlators. We use Japanese fermionic technique. This work is a continuation of joint works with John Harnad \cite{HO1,HO2,HO3}.

\section{Introduction}

This work is a continuation of \cite{HO1,HO2,HO3,O-2021} which
concerned bilinear expansions of Schur lattices $\{\pi_{\lambda}(g) ({\bf t})\}$ of KP
$\tau$-functions, labeled by partitions
$\lambda$ and
$\GL(\infty)$ group elements $\hat{g}$
expressible as sums of products of corresponding lattices $\{\kappa_{\alpha} (h)({\bf t}_B)\}$ of BKP $\tau$ functions,
labeled by strict partitions $\alpha$ and $\SO(\infty)$ group elements $\hat{h}$.

The approach was based on the notion of tau function, as introduced by Sato \cite{Sa} expressed
as vacuum state expectation values (VEV's) of products free fermionic operators, and their
exponentials, as developed in the works of Kyoto school \cite{DJKM1,DJKM2,JM}.

\subsection{Few words about the problem}

Definitions will be given later below in Section \ref{Preliminaries}.

{\it $(a)$ Schur functions and the projective Schur functions.}
Let $\alpha=(\alpha_1,\dots,\alpha_r)$ is a strict partition and
$D(\alpha):=(\alpha_1,\dots,\alpha_r|\alpha_1-1,\dots,\alpha_r-1)$ and $\hat r =r$ if $r$ is even
and $\hat r=r+1$ if~$r$~is odd.
 The following relation is known
(see, for instance, \cite{Mac1}):
\be\label{Q^2=s}
2^{-\hat r}\bigl(Q_\alpha\bigl({\pb^{\rm B}}\bigr) \bigr)^2 = s_{D(\alpha)}(\pb),
\ee
where $Q_\alpha$ is the projective Schur function written as the function of power sum variables
\be
\label{p_B}
{\pb^{\rm B}}=\bigl(p^{\rm B}_1,p^{\rm B}_3,p^{\rm B}_5,\dots\bigr)
\ee
and where $s_\lambda$ is the Schur function written as function of power sum variables
\begin{gather*}
\pb=(p_1,p_2,p_3,p_4,p_5,\dots),
\end{gather*}
and in formula (\ref{Q^2=s}) these two sets of variables are related as follows:
\be\label{p'}
\pb = \pb' :=(p_1,0,p_3,0,p_5,0,\dots),
\ee
where
\be\label{p_p^B}
p_{i}=2p^{\rm B}_{i},\qquad i\ \text{odd},
\ee
where $p^{\rm B}_n$ and~$p_n$ are commonly related to the two different sets of the variables,
say $z_1,\dots,z_N$ and $x_1,\dots,x_M$, respectively, as Newton sums
\begin{gather}\label{power_sum_via_z}
p^{\rm B}_{2m-1}=p_{2m-1}(z_1,z_2,\dots)=\sum_{i=1}^{N} z_i^{2m-1},\qquad m=1,2,3,\dots,
\\
p_m=p_m(x_1,\dots,x_M)=\sum_{i=1}^M x_i^m,\qquad m=1,2,3,\dots, \nonumber
\end{gather}
and therefore are called power sum variables.

 The projective Schur function
is a polynomial quasihomogeneous function in the power sum variables and symmetric homogeneous polynomial
in the variables $\zb = (z_1,\dots,z_{N(\zb)})$, related to the power sums by (\ref{power_sum_via_z}).
It is labeled by a strict partition (multiindex)
$\alpha=(\alpha_1,\dots,\alpha_k)$ (where $\alpha_1>\alpha_2>\cdots >\alpha_k \ge 0$ is the set of integers).
 We~will recall the definition later in the text.
We~will write the projective Schur function either as symmetric function~$Q_\alpha(\zb)$ or as polynomial~$Q_\alpha\bigl({\pb^{\rm B}}\bigr)$ and we hope it does not produce a misunderstanding.

The function $s_{D(\alpha)}(\pb)$ in the right-hand side is the Schur function
labeled by a special partition (the multiindex) denoted by $D(\alpha)$ which is called the double of
the strict partition $\alpha$. In general, the Schur function $s_\lambda(\pb) $ is defined for {\em any}
partition $\lambda$
and is a quasihomogeneous polynomial in the variables $(p_1,p_2,p_3,\dots)=:\pb$, where in contrast to (\ref{p'})
the variables indexed with even numbers are also in presence. However the Schur function in the right-hand side
of (\ref{Q^2=s}) is evaluated for the restricted
set of variables denoted by $\pb'$ which is $(p_1,0,p_3,0,p_5,0,\dots)$ and~$\lambda$ is chosen to be the double of
$\alpha$. In what follows, we also use the so-called Frobenius coordinated of $\lambda$,
$\lambda=(\alpha|\beta)$, where $\alpha=(\alpha_1,\dots,\alpha_r)$ and $\beta=(\beta_1,\dots,\beta_r)$ is the pair
of strict partitions. The function $s_\lambda$ is a symmetric homogeneous polynomial in the variables $\xb$
and commonly is written as $s_\lambda(\xb)$ which is $s_\lambda(\pb(\xb))$ (we hope it will not produce the
incomprehension).

Equality (\ref{Q^2=s}) is the fundamental equality which relates Schur functions and projective Schur functions,
to our knowledge at first it was found in \cite{You}.

Consider two independent sets of variables $\ab=(a_1,\dots,a_N)$ and $\bb=(b_1,\dots,b_N)$.
(Sometimes we will use the notation $N(\ab)$ for the number of the variables in the set $\ab$.)

Let us choose
\be\label{p=p(x,y)}
p_n = p_n(\ab/{-}\bb):=\sum_{i=1}^N \left(a_i^n-(-b_i)^n \right),\qquad n=1,2,3,\dots .
\ee
Such $\pb(\ab/{-}\bb)$ we will call {\em supersymmetric Newton sums} of the variables $\ab$, $\bb$.
Below, we use the notation $\pb(\ab/{-}\bb)=\left( p_1(\ab/{-}\bb),p_2(\ab/{-}\bb), p_3(\ab/{-}\bb),\dots\right)$
and $s_\lambda(\ab/{-}\bb):=s_\lambda(\pb(\ab/{-}\bb))$ (the similar notation is used in~\cite[Chapter~I, Section~3, Example~23]{Mac1}).

Notice that $s_\lambda(\ab/{-}\ab)$ can be written as $s_\lambda(\pb')$, where $p'_{2n-1}=2\sum_{i=1}^N a_i^{2n-1}$.
We will see that the ``sypersymmetric Newton sums'' are quite natural for our problems and we shall use it throughout
the paper.

\begin{rem}
One can easily conjecture that for any set $\pb=(p_1,p_2,\dots)$ there exits such a~number $N$ (perhaps, infinite) and
two sets $\ab=(a_1,\dots,a_N) $ and $\bb=(b_1,\dots,b_N)$ that (\ref{p=p(x,y)}) is~true.
\end{rem}

In the present work, we find the generalization of (\ref{Q^2=s}) which symbolically may be written as%
\be
\label{main_example}
s_{(\alpha|\beta)}(\ab/{-}\bb)=2^{-\hat r}\sum_{(\zeta^+,\zeta^-)\in \PP(\alpha,\beta)
\atop (\zb^+,\zb^-)\in \PP(\ab,\bb)}
\left[ \zeta^+,\zeta^-\atop \alpha,\beta \right]\left[\zb^+,\zb^-\atop \ab,\bb\right]
Q_{\zeta^+}(\zb^+)Q_{\zeta^-}(\zb^-) ,
\ee
where the partition $\lambda=(\alpha|\beta)$ and the sets of complex numbers $\ab$, $\bb$ are free.
The sum in the right-hand side is taken over splittings of the set of the Frobenius coordinated
$\alpha\cup(\beta+1)$ into ordered subsets $\zeta^-$ and $\zeta^+$, and over splittings of the set
of the coordinates $\ab$, $\bb$ into ordered subsets $\zb^-$ and $\zb^+$,
 The weights denoted by square
brackets will be written down in~(\ref{a}) and~(\ref{d*}),
see Section \ref{notations-we-need}.
The projective Schur functions in the right-hand side are written as symmetric functions in the variables
$\zb^\pm=\bigl(z^\pm_1,\dots,z^\pm_{N(\zb^\pm)}\bigr)$ selected from the set~$a_1,\dots,a_N,b_1,\dots,b_N$. The parts of the partitions $\zeta^\pm=\bigl(\zeta^\pm_1,\dots,\zeta^\pm_{m(\zeta^\pm)}\bigr)$
 are selected from the set
$\alpha_1,\dots,\alpha_r,\beta_1+1,\dots,\beta_r+1$.\footnote{It would be symmetric for semiinteger partitions where $\alpha_i\to \alpha_i+\frac 12$,
$\beta_i\to \beta_i-\frac 12$ and the labels of the Fourier modes $\psi_i\to \psi_{i+\frac 12}$,
as it was done in Kac's papers,
but we avoid semiintegers. }

Actually, the case $\ab=\bb$ was already studied in our previous work \cite{HO1}.
Let us note that under the pair of restrictions: $\alpha=\beta+1$ and $\ab=\bb$
we get formula (\ref{Q^2=s}) as the particular case of (\ref{main_example}).

{\it $(b)$ KP and BKP lattice tau functions.}
Apart of the relation (\ref{main_example}) we get much more general relation, that is a relation
between KP and BKP tau functions.
Let us remind that the Schur function is the simplest nontrivial example of the KP tau function \cite{Sa} while the
projective Schur function is the simplest nontrivial example of the BKP tau function \cite{Nim,You}.\footnote{Let us note that by BKP tau function we mean the tau function introduced by Kyoto school in \cite{DJKM1}.
Another tau function also called the BKP one was introduced in \cite{KvdL1}. Both tau functions have
a lot of applications in various problems of mathematical physics, for instance, in random matrix theory.}

The lattice KP tau function can be written as a sum over partitions as follows:
\be\label{lattice_KP_series}
S_\lambda(\pb|\hat{g})=\sum_{\mu\in\Pa} s_\mu(\pb) \hat{g}_{\mu,\lambda}
\ee
and depends on $\pb=(p_1,p_2,\dots)$ and also on a partition (a multiindex) $\lambda\in\Pa$, where
$\Pa$ is the set of all partitions.

Below we imply that $\pb=\pb(\ab/{-}\bb)$.
The lattice BKP tau function can be written as
\be\label{lattice_BKP_series}
K_\mu\bigl(\zb|\hat{h}\bigr)=\sum_{\nu\in\DP} Q_\nu(\zb) \hat{h}_{\nu,\mu} .
\ee
In these formulas, $\hat{g}_{\mu,\lambda}$ and $\hat{h}_{\nu,\mu}$ are certain coefficients
given by the choice of KP and BKP tau functions, see Section \ref{2-sided,lattice} below.
This choice may be treated as a choice of a certain element of the Clifford group $\hat{g}$ in the KP case
which is defined by a choice of $\hat{g}\in\widehat{{\rm GL}}_\infty$ in the KP case
(see~(\ref{g=exp}) below) and
by $\hat{h}\in \hat{B}_\infty$ in the BKP case (formula (\ref{h=exp}) below) (therefore, we label the left-hand sides of
(\ref{lattice_KP_series}) and (\ref{lattice_BKP_series}) with these symbols).

\begin{rem}\label{smallSchur}\quad
\begin{itemize}\itemsep=0pt
\item If $\hat{g} = 1$, then $S_\lambda(\pb|\hat{g}=1)=s_\lambda(\pb)$.

\item If $\hat{h} = 1$, then
$K_\mu\bigl(\zb|\hat{h}=1\bigr)=Q_\mu(\zb)$.
\end{itemize}
\end{rem}
\begin{rem}\quad
\begin{itemize}\itemsep=0pt
\item If $\lambda=0$, then $S_{\lambda =0}(\pb|\hat{g})$ is usual (``one-side'') KP tau function.
\item If $\mu=0$, then $K_{\mu=0}\bigl({\pb^{\rm B}}|\hat{h}\bigr)$ is usual (``one-side'') BKP tau function.
\end{itemize}
\end{rem}

 In the present work, it is supposed that
$\hat{g},\hat{h}^\pm\in \hat{B}_\infty$ and
$\hat{g}=\hat{h}^-\hat{h}^+$.

Under parametrization (\ref{p=p(x,y)}) and the mentioned condition
$\hat{g}=\hat{h}^-\hat{h}^+\in \hat{B}_\infty$
explained in~Section~\ref{2-sided,lattice},
the relation between (\ref{lattice_KP_series}) and (\ref{lattice_BKP_series})
is identical to the relation (\ref{main_example}), where $s_\lambda(\ab/{-}\bb)$ is replaced by
$S_\lambda(\pb(\ab/{-}\bb)|\hat{g})$
of (\ref{lattice_KP_series}) and where $Q_{\zeta^\pm}\bigl(\zb^\pm\bigr) $ are replaces by $K_\mu\bigl(\zb |\hat{h}^\pm\bigr)$ of
(\ref{lattice_BKP_series}):
\be\label{main_example_dressed}
S_{(\alpha|\beta)}\bigl(\ab/{-}\bb|\hat{h}^+\hat{h}^-\bigr)=\sum_{ (\zeta^+,\zeta^-)\in \PP(\alpha,\beta)
\atop (\zb^+,\zb^-)\in \PP(\ab,\bb)}
\left[ \zeta^+,\zeta^-\atop \alpha,\beta\right]\left[\zb^+,\zb^-\atop \ab,\bb\right]
K_{\zeta^+}\bigl(\zb^+|\hat{h}^+\bigr)K_{\zeta^-}\bigl(\zb^-|\hat{h}^-\bigr) ,
\ee
where the summation range
and the weights denoted by square
brackets are the same as in (\ref{main_example}).
This is a subject of Theorem \ref{lattice-tau}.

\begin{rem}
In the spirit of the terminology common in soliton theory, one may say that relation (\ref{main_example_dressed})
is the {\em dressed} by element $\hat{g}$ version of (\ref{main_example}).
\end{rem}

\begin{rem}\label{notational_remark}
In our previous work \cite{HO2}, we use different notations
\begin{gather*}
\pi_{(\alpha|\beta)}(\hat{g}) ({\pb})=S_{(\alpha|\beta)}(\pb|\hat{g}),
\qquad
\kappa_{\alpha} \bigl(\hat{h}\bigr)\bigl(2{\pb^{\rm B}}\bigr)=K_\alpha\bigl(2{\pb^{\rm B}}|\hat{h}\bigr)
\end{gather*}
see Remarks \ref{Plucker} and \ref{Cartan}.
\end{rem}

\begin{rem}
We call a KP lattice tau function {\em polynomial} in case for any $\lambda$, in (\ref{lattice_KP_series}), there is
only a finite number of terms in the right-hand side. Examples of the polynomial KP tau functions
were presented in \cite{HL, HO3}.
Similarly, we call a BKP lattice tau function polynomial if there is a finite number of terms
in the right-hand side of (\ref{lattice_BKP_series}) for any $\mu$. Polynomial BKP tau functions
were studied in \cite{KvdL2,KvdLRoz}.

As we mentioned (see Remark \ref{smallSchur}),
the simplest example of the KP polynomial tau function is the Schur function $s_\lambda(\pb)$. Other examples,
like characters of linear groups or Laguerre polynomials can be found in \cite{HL, HO3}, respectively.
The simplest example of the BKP polynomial tau function is the projective $Q_\mu$-function. Other examples may be
found in \cite{HO3} and in the references therein.
\end{rem}

{\it $(c)$ The relation between KP and BKP two-sided tau functions.}
A two-sided KP tau function can be written as a double sum over partitions
\be\label{2_sided_series}
\tau(\pb,\tilde{\pb}|\hat{g})=
\sum_{\lambda,\mu} s_\mu(\pb) \hat{g}_{\mu,\lambda} s_\lambda(\tilde{\pb}) .
\ee
Each two-sided KP tau function depends on two sets
of higher times $\pb=(p_1,p_2,\dots)$ and $\tilde{\pb}=(\tilde{p}_1,\tilde{p}_2,\dots)$.
In the present paper we consider only the case $\pb=\pb(\ab/{-}\bb)$ according to~(\ref{p=p(x,y)}).
Similarly, we choose $\tilde{\pb}$ to be
\be\label{p=p(x,y)_tilde}
\tilde{p}_n=\tilde{p}_n\bigl(\tilde{\ab}/{-}\tilde{\bb}\bigr)=
\sum_{i=1}^{\tilde{N}}\bigl( \tilde{a}_i^n -\bigl(-\tilde{b}_i\bigr)^n \bigr) .
\ee
We will write two-sided KP tau function as $\tau\bigl(\ab/{-\bb},\tilde{\ab}/{-\tilde{\bb}}|\hat{g}\bigr)$ for the choice
(\ref{p=p(x,y)}) and (\ref{p=p(x,y)_tilde}).

The solution of these equations is defined
by the choice of $\hat{g}$ which gives rise to the coefficients~$\hat{g}_{\mu,\lambda}$,
in a way described in \cite{Takasaki_Schur}.
The alternative description uses the notation of Sato Grassmannian and to speak about a point of the Grassmannian
instead of $\hat{g}$.
We will try to avoid this notion in order not to convert a given clear problem of writing down an explicit equality
to a~part of geometry.\footnote{The geometrical approach of the related topics can be found in \cite{BHH}.}

The two-sided BKP tau function can be written as a double sum over strict partitions
(definitions see below):
\be\label{2_sided_series_B}
\tau^{\rm B}\bigl({\pb^{\rm B}},\tilde{\pb}^{\rm B}|\hat{h}\bigr)=
\sum_{\lambda,\mu} 2^{-\frac12 \ell(\mu)-\frac12\ell(\lambda)}Q_\mu\bigl({\pb^{\rm B}}\bigr) \hat{h}_{\mu,\lambda} Q_\lambda\bigl(\tilde{\pb}^{\rm B}\bigr) .
\ee
It depends on two sets of higher times labeled with odd numbers: ${\pb^{\rm B}}=\bigl(p^{\rm B}_1,p^{\rm B}_3,\dots\bigr)$ and
$\tilde{\pb}^{\rm B}=\bigl(\tilde{p}^{\rm B}_1,\tilde{p}^{\rm B}_3,\dots\bigr)$. In the case, for sets
$\zb=\bigl(z_1,\dots,z_{N(\zb)}\bigr)$ and $\tilde{\zb}=\bigl(z_1,\dots,z_{N(\tilde{\zb})}\bigr)$, we have
\be\label{p_B_tilde}
p^{\rm B}_n=p^{\rm B}_n(\zb)=\sum_{i=1}^{N(\zb)} z^n_i,\qquad
\tilde{p}^{\rm B}_n=\tilde{p}^{\rm B}_n(\tilde{\zb})=\sum_{i=1}^{N(\tilde{\zb})} \tilde{z}^n_i .
\ee
We will write
$\tau^{\rm B}\bigl({\pb^{\rm B}}(\zb),{\pb^{\rm B}}(\tilde{\zb})|\hat{h}\bigr) =:\tau^{\rm B}\bigl(\zb,\tilde{\zb}|\hat{h}\bigr)$.

If we compare (\ref{lattice_KP_series}) with (\ref{2_sided_series}) (also (\ref{lattice_BKP_series})
with (\ref{2_sided_series_B})), we see that it is a sort of Fourier transform.

The relation between two-sided KP and two-sided BKP tau functions is as follows:
\begin{gather}
\tau\bigl(\ab/{-}\bb,\tilde{\ab}/{-}\tilde{\bb}|\hat{h}^+\hat{h}^-\bigr)\nonumber\\
\qquad{}=
\sum_{(\zb^+,\zb^-)\in \PP(\ab,\bb)\atop(\tilde{\zb}^+,\tilde{\zb}^-)\in\PP(\tilde\ab,\tilde\bb)}
\left[\zb^+,\zb^-\atop \ab,\bb\right]\left[\tilde{\zb}^+,\tilde{\zb}^-\atop \tilde{\ab},\tilde{\bb}\right]^*
\tau^{\rm B}\bigl(\zb^+,\tilde{\zb}^+|\hat{h}^+\bigr) \tau^{\rm B}\bigl(\zb^-,\tilde{\zb}^-|\hat{h}^-\bigr).\label{relation_c}
\end{gather}
The sum in the right-hand side is taken over the natural numbers $N(\zb^-),N(\tilde{\zb}^-)$ and over the sets
of (complex) numbers
$
\zb^\pm =\bigl(z^\pm_1,\dots,z^\pm_{N(\zb^\pm)}\bigr)$, $\zb^\pm =\bigl(\tilde{z}^\pm_1,\dots,z^\pm_{N(\tilde{\zb}^\pm)}\bigr)$,
where $N(\zb^+)+N(\zb^-)=2N$, $N\bigl(\tilde{\zb}^+\bigr)+N(\tilde{\zb}^-)=
2\tilde{N}$, and
implies that $\zb^+\cup\zb^-=\ab\cup\bb$
and $\tilde{\zb}^+\cup\tilde{\zb}^-=\tilde{\ab}\cup\tilde{\bb}$.
The coefficients denoted
by square brackets will be written down in (\ref{d*}) and in (\ref{d}).

The simplest and trivial example of relation (\ref{relation_c}) is the case where
$\hat{g}_{\mu,\lambda}=\delta_{\mu,\lambda}$ (the case $\hat{g}=1$) and $\ab=\bb$, $\tilde{\ab}=\tilde{\bb}$.
In the power sum variables, it is written as
\begin{align}
\sum_{\lambda} s_\lambda(\pb')s_\lambda(\tilde{\pb}')&{}=
{\rm e}^{\sum_{n>0,\,\text{odd}}\frac{1}{n}p_n\tilde{p}_n}=\bigl({\rm e}^{\sum_{n>0}\frac {2}{n} p^{\rm B}_n\tilde{p}_n^{\rm B} } \bigr)^2 \nonumber\\
&{}=\biggl(\sum_{\mu} 2^{-\ell(\mu)} Q_\mu\bigl({\pb^{\rm B}}\bigr)Q_\mu\bigl(\tilde{\pb}^{\rm B}\bigr)\biggr)^2 ,\label{vacuum}
\end{align}
where the sets $\pb'$ and ${\pb^{\rm B}}$ are the same as in (\ref{p_B}), (\ref{p'}) and (\ref{p_p^B}).
In the variables $\ab$, $\tilde\ab$, both sides of (\ref{vacuum}) are equal to
\[
\Biggl( \prod_{N\ge i>j} \frac{a_i-a_j}{a_i+a_j}
\prod_{\tilde{N}\ge i>j} \frac{\tilde{a}_i-\tilde{a}_j}{\tilde{a}_i+\tilde{a}_j}\Biggr)^2 .
\]

{\it $(d)$ Bi-lattice KP and bi-lattice BKP tau functions.}
In view of the notion of the lattice tau functions (\ref{lattice_KP_series}) and (\ref{lattice_BKP_series}),
it is natural to call $\hat{g}_{\lambda,\tilde{\lambda}}$ bi-lattice KP tau function and to call
$\hat{h}_{\mu,\tilde{\mu}}$ bi-lattice BKP tau function. Then we get
\[
\hat{g}_{(\alpha|\beta),(\tilde{\alpha}|\tilde{\beta})}=
\sum_{ (\zeta^+,\zeta^-)\in \PP(\alpha,\beta)\atop
 (\tilde\zeta^+,\tilde\zeta^-)\in \PP(\tilde\alpha,\tilde\beta) }
\left[ \zeta^+,\zeta^-\atop \alpha,\beta\right]\left[\tilde{\zeta}^+,\tilde{\zeta}^-\atop \tilde{\alpha}, \tilde{\beta} \right]
h^+_{\zeta^+,\tilde{\zeta}^+}h^-_{\zeta^-,\tilde{\zeta}^-} ,
\]
where $\zeta^+\cup\zeta^-=\alpha\cup(\beta+1)$ and $\tilde{\zeta}^+\cup\tilde{\zeta}^-=\tilde{\alpha}\cup\bigl(\tilde{\beta}+1\bigr)$
and where $\lambda=(\alpha|\beta)$, $\tilde{\lambda}=\bigl(\tilde{\alpha}|\tilde{\beta}\bigr)$,
$\alpha=(\alpha_1,\dots,\alpha_r)$, $\beta=(\beta_1,\dots,\beta_r)$,
 $\tilde{\alpha}=(\tilde{\alpha}_1,\dots,\tilde{\alpha}_{\tilde{r}})$,
 \smash{$\tilde{\beta}=(\tilde{\beta}_1,\dots,\tilde{\beta}_{\tilde{r}})$}.
The weights denoted by square
brackets will be written below, see (\ref{a}).

The simple nontrivial example of the bi-lattice KP tau function $\hat{g}_{\lambda,\tilde{\lambda}}$ is the product
$s^*_\mu(\lambda)\!\dim\lambda$, where $s^*_\mu(\lambda)$ is the so-called shifted Schur function
introduced by Okounkov in \cite{EOP} and $\dim\lambda$ is the number of standard tableaux of the shape
$\lambda$, see \cite{Mac1}. The shifted Schur functions were used in an approach to the representation theory
developed by G. Olshanski and A. Okounkov in~\cite{OO}.

The simple nontrivial example of the bi-lattice BKP tau function $h_{\mu,\nu}$
is the following product $Q^*_\mu(\nu) \dim^{\rm B}\mu$, where $Q^*_\mu(\nu)$ is the shifted projective Schur function
introduced by Ivanov in \cite{Iv} and $\dim^{\rm B}\mu$ is the number of the shifted standard tableaux of
the shape $\mu$. Functions $Q^*_\mu(\nu)$ are of use in the study of spin Hurwitz numbers \cite{MMN2019}.

The relation between shifted Schur functions and the shifted projective Schur functions was written
done in \cite{O-2021}.

\section{Preliminaries}\label{Preliminaries}

Here we review known basic facts and certain previous results and introduce notations.

\subsection{Some notations we use (partitions and ordered sets)}\label{notations-we-need}

{\bf Partitions.}
We recall that a nonincreasing set of nonnegative integers $\lambda_1\ge\cdots \ge \lambda_{k}\ge 0$,
we call partition $\lambda=(\lambda_1,\dots,\lambda_{l})$, and $\lambda_i$ are called parts of $\lambda$.
The sum of parts is called the weight~$|\lambda|$ of~$\lambda$. The number of nonzero parts of $\lambda$
is called the length of $\lambda$, it will be denoted by $\ell(\lambda)$, see~\cite{Mac1} for details.
Partitions will be denoted by Greek letters: $\lambda,\mu,\dots$. The set of all partitions is denoted by
$\Pa$. The set of all partitions with odd parts is denoted by $\OP$.
Partitions with distinct parts are called strict partitions, we prefer
letters $\alpha$, $\beta$, $\zeta$ \big(also $\zeta^\pm$\big) to denote them.
The set of all strict partitions will be denoted by $\DP$.
The Frobenius coordinated $\alpha$, $\beta$ for partitions $(\alpha|\beta)=\lambda\in\Pa$ are of usernames.
Let me recall that the coordinates $\alpha=(\alpha_1,\dots,\alpha_r)\in\DP$ consists of the lengths of arms counted
from the main diagonal (the diagonal nodes are not included) of the Young diagram of $\lambda$ while
$\beta=(\beta_1,\dots,\beta_r)\in\DP$ consists of the lengths of legs counted
from the main diagonal (again without the nodes on the main diagonal) of the Young diagram of $\lambda$, $r$
is the length of the main diagonal of $\lambda$, we call it Frobenius rank.

For example, the partition $(1)$
in Frobenius coordinate is written as $(0|0)$ and its Frobenius rank is 1. Other examples: the partitions
$(1,1,1)$, $(n)$ and $(n,m)$, $m\ge 2$ in the Frobenius coordinates are written as $(0|2)$ (rank equal to~1),
$(n-1|0)$ (rank equal to~1) and $(n-1,m-2|1,0))$ (rank equal to~2), respectively.
See \cite{Mac1} for details about partitions and their Young diagrams.

The number of the nonvanishing part of a partition $\lambda$ is called the length of $\lambda$ and is denoted by
$\ell(\lambda)$. A partition
also can be presented as $\lambda=(1^{m_1}2^{m_2}3^{m_3}\cdots)$, where $m_i$ is the number of time the number
$i$ occurs in the partition. We shall use the notation
\be\label{z-lambda}
z_\lambda =\prod_i m_i!i^{m_i}.
\ee

\begin{de}[supplemented partitions]\label{supplemented-partition}
If $\zeta$ is a strict partition of cardinality $m(\zeta)$ (with $0$ allowed as a part),
we define the associated {\em supplemented partition} $\hat{\zeta}$ to be
\begin{gather*}
\hat{\zeta} := \begin{cases} \zeta & \text{if } m(\zeta) \ \text{is even}, \\
 (\zeta,0) & \text{if } m(\zeta) \ \text{is odd}.
 \end{cases}
\end{gather*}
We denote by $m\big(\hat{\zeta}\big)$ the cardinality of $\hat{\zeta}$ which is an even number.
\end{de}
Let us note that the supplemented partition is not necessarily strick. For instance,
if $\zeta=0$, then $\hat\zeta=(0,0)$.

{\bf Polarization $\boldsymbol{\PP(\alpha,\beta)}$ of the set of the Frobenius coordinates $\boldsymbol{\alpha}$, $\boldsymbol{\beta}$ \cite{HO2}.}
This paragraph is taken from \cite{HO2}.
Consider a partition $\lambda$ written in the Frobenius coordinate as $(\alpha|\beta)=(\alpha_1,\dots,\alpha_r|\beta_1,\dots,\beta_r)$, $\alpha_1>\cdots >\alpha_r\ge 0$,
$\beta_1>\cdots >\beta_r\ge 0$.

\begin{de}
A {\em polarization} of $(\alpha|\beta)$, is a pair $\zeta:=\bigl(\zeta^+, \zeta^-\bigr)$
of strict partitions with cardinalities (or {\em lengths})
\begin{gather*}
m\bigl(\zeta^+\bigr):= \#\bigl(\zeta^+\bigr), \qquad m(\zeta^-)):= \#(\zeta^-)
\end{gather*}
(including possibly a zero part $\zeta^+_{m(\zeta^+)}=0$ or $\zeta^-_{m(\zeta^-)}=0$), satisfying
\begin{gather*}
\zeta^+ \cap \zeta^- = \alpha \cap I(\beta), \qquad \zeta^+ \cup \zeta^- = \alpha \cup I(\beta),
\end{gather*}
where
\begin{gather*}
I(\beta) := (I_1(\beta), \dots, I_r(\beta))
\end{gather*}
is the strict partition \cite{Mac1} with parts
\begin{gather*}
I_j(\beta) = \beta_j +1, \qquad j=1, \dots, r.
\end{gather*}
The set of all polarizations of $(\alpha|\beta)$ is denoted by $\PP(\alpha,\beta)$.
\end{de}
We denote the strict partition obtained by intersecting $\alpha$ with $I(\beta)$ as
\begin{gather*}
S := \alpha \cap I(\beta)
\end{gather*}
and its cardinality as
\[
s := \#(S).
\]
Since both $\alpha$ and $I(\beta)$ have cardinality $r$, it follows that
\begin{gather*}
m\bigl(\zeta^+\bigr) + m(\zeta^-) = 2r,
\end{gather*}
so $m\bigl(\zeta^\pm\bigr)$ must have the same parity. It is easily verified \cite{HO1} that the cardinality
of $\PP(\alpha, \beta)$ is~$2^{2r -2s}$.
The following was proved in \cite{HO1}.

\begin{lem}[binary sequence associated to a polarization]
For every polarization $\zeta:=\bigl(\zeta^+, \zeta^-\bigr)$ of $\lambda=(\alpha| \beta)$, there is a unique binary sequence
of length $2r$
\begin{gather*}
\epsilon(\zeta) =(\epsilon_1(\zeta), \dots, \epsilon_{2r}(\zeta)),
\end{gather*}
with
$\epsilon_j(\zeta) =\pm$, $j=1, \dots 2r$,
such that
\begin{enumerate}\itemsep=0pt
\item[$(1)$] The sequence of pairs
\be
(\alpha_1, \epsilon_1(\zeta)), \dots, (
\alpha_r, \epsilon_r(\zeta)), (\beta_1+1, \epsilon_{r+1}(\zeta)), \dots , ( \beta_r+1, \epsilon_{2r}(\zeta))
\label{alpha_beta_sequence}
\ee
is a permutation of the sequence
\be
\bigl(\zeta^+_1,+\bigr), \dots, \bigl( \zeta^+_{m^+(\zeta)},+\bigr), (\zeta^-_1, -), \dots , (\zeta^-_{m^-(\zeta)}, -\bigr).
\label{mu_sequence}
\ee
\item[$(2)$]
$\epsilon_j(\zeta)= +$ if $\alpha_j \in S$, and $\epsilon_{r+j}(\zeta)=
 -$ if $\beta_j +1\in S$, $j=1, \dots, r.
$
\end{enumerate}
\end{lem}

\begin{de}
 The {\em sign} of the polarization $\bigl(\zeta^+, \zeta^-\bigr)$, denoted by $\sgn(\zeta)$, is defined
as the sign of the permutation that takes the sequence (\ref{alpha_beta_sequence}) into the sequence
(\ref{mu_sequence}).
\end{de}
Denote by
\begin{gather*}
\pi\bigl(\zeta^\pm\bigr) :=\#\bigl(\alpha \cap \zeta^{\pm}\bigr)
\end{gather*}
the cardinality of the intersection of $\alpha$ with $\zeta^{\pm}$. It follows that
\begin{gather*}
\pi\bigl(\zeta^+\bigr) + \pi(\zeta^-) = r+s.
\end{gather*}
Now we introduce the notation
\be
\label{a}
 \left[\zeta^+,\zeta^-\atop \alpha,\beta \right]:={(-1)^{\frac{1}{2}r(r+1) + s}\over 2^{r-s}}
\sgn(\zeta)(-1)^{\pi(\zeta^-)+\frac12 m(\hat{\zeta}^-)} ,
\ee
where $r$ is the Frobenius rank of $(\alpha|\beta)$.

{\bf The ordered coordinate sets.}
Consider given sets $\ab$, $\bb$ of complex numbers
\be\label{Fixed_sets}
\ab=(a_1,\dots,a_N), \qquad \bb=(b_1,\dots,b_N).
\ee
We want to have the similarity of the sets of complex numbers with partitions.
For this purpose, let us introduce the following order in the set $\ab\cup\bb$: in the set called ordered
we place complex number as a set with the weakly decaying set with respect to the absolute values
of the numbers. In case the absolute values in a subset is the same, we just fix any order inside the subset
and keep it in what follows. Such sets we will call time ordered sets.\footnote{We recall that the complex
coordinate, say $z$ of a quantum field in Euclidean 2D theory is presented as ${\rm e}^{\sqrt{-1}\varphi -\tau}$,
where $\tau$ is interpreted as time variable.}
Let us also treat the subsets~(\ref{Fixed_sets}) $\ab$ and~$\bb$ of the set $\ab\cup\bb$ as also ordered:
$|a_1|\ge \cdots \ge |a_N|$ and $|b_1|\ge\cdots \ge |b_N|$.

Starting from now, all coordinate sets will be treated as the ordered sets whose each member is labeled in
the appropriate way (labels goes up from the left to the right). One can consider another pair of the ordered
complementary subsets, let us denote it as $\zb^+$ and $\zb^-$, $\zb^+\cup\zb^-=\ab\cup\bb$
(the order inherits the given order of the set $\ab\cup\bb$)
whose cardinalities $N(\zb^+)$ and $N(\zb^-)$
are not necessarily equal to $N$:
\be\label{current_sets}
\zb^\pm=\bigl(z^\pm_1,\dots,z^\pm_{N(\zb^\pm)}\bigr),\qquad \bigl|z^\pm_i\bigr|\ge \bigl|z^\pm_{i+1}\bigr|, \qquad N(\zb^+)+N(\zb^-)=2N .
\ee
In what follow, sets (\ref{Fixed_sets}) and (\ref{current_sets}) will play different roles:
the sets (\ref{Fixed_sets}) and the given number~$N$ are constant along the paper
while the sets $\zb^+$ and $\zb^-$ will be varying and the numbers $N\bigl(\zb^\pm\bigr)$ can be summation indices.

On the given sets (\ref{Fixed_sets}) and also on sets (\ref{current_sets}) let us introduce
the following involution $*$:
\be\label{x*y*}
a^*_i:=-a^{-1}_{N-i},\qquad b^*_i:=-b^{-1}_{N-i},\qquad z^{\pm*}_i:=-\bigl(z^{\pm}_{N(\zb)-i}\bigr)^{-1} .
\ee
Then the sets
$
\ab^*=(a^*_1,\dots,a^*_N)$, $\bb^*=(b^*_1,\dots,b^*_N)$, $\zb^{\pm*}=\bigl(z^{\pm*}_1,\dots,z^{\pm*}_{N(\zb^\pm)}\bigr)$
are also ordered in the sense that $|a^*_i|\ge|a^*_{i+1}|$,
$|b^*_i|\ge|b^*_{i+1}|$, $\bigl|z^{\pm*}_i\bigr|\ge\bigl|z^{\pm*}_{i+1}\bigr|$ and the set
$\ab^*\cup\bb^*$ is ordered.

By analogy with Definition \ref{supplemented-partition}, we need the following.

\begin{de}[supplemented set of coordinates $\zb^\pm$, $\zb^{\pm*}$]\label{supplemented-coordinate-sets}
If $\zb^\pm$ is a strict partition of cardinality $N\bigl(\zb^\pm\bigr)$ (with $0$ allowed as a part),
we define the associated {\em supplemented set} $\hat{\zb}^\pm$ to~be%
\begin{gather*}
\hat{\zb}^\pm := \begin{cases} \zb^\pm & \text{if } N\big(\zb^\pm\big) \ \text{is even}, \\
 \big(\zb^\pm,0\big) & \text{if } N\big(\zb^\pm\big) \ \text{is odd},
 \end{cases}\qquad
 \hat{\zb}^{\pm*} := \begin{cases} \zb^{\pm*} & \text{if } N\big(\zb^\pm\big) \ \text{is even}, \\
 \big(\zb^{\pm*},\infty\big) & \text{if } N\big(\zb^{\pm*}\big) \ \text{is odd}.
 \end{cases}
\end{gather*}
We denote by $N\bigl(\hat{\zb}^\pm\bigr)$ and $N\bigl(\hat{\zb}^{\pm*}\bigr)$ the cardinality of respectively $\hat{\zb}^\pm$
and $\hat{\zb}^{\pm*}$.
We get $N\bigl(\hat{\zb}^\pm\bigr)=N\bigl(\hat{\zb}^{\pm*}\bigr)$ is an even number.
\end{de}

{\bf Vandermond-like products.}
For given sets
$\ab=(a_1,\dots,a_N)$, $\zb^\pm=\bigl(z^\pm_1,\dots,z^\pm_{N(\zb^\pm)}\bigr)$,
we use the following notations:
\begin{gather*}
\Delta(\ab):=\prod_{i<j\le N}(a_i-a_j) ,\qquad
\Delta^{\rm B}\bigl(\zb^\pm\bigr):=\prod_{i<j\le N(\zb^\pm)}\frac{z^\pm_i-z^\pm_j}{z^\pm_i+z^\pm_j} ,
\\
\Delta(\ab/{-}\bb):=\frac{\Delta(\ab)\Delta(\bb)}{\prod_{i,j=1}^N(a_i+b_j)}=
\det \bigl((a_i+b_j)^{-1}\bigr)_{i,j=1,\dots,N}.
\end{gather*}

As we see
$\Delta^{\rm B}\bigl(\hat{\zb}^\pm\bigr)=\Delta^{\rm B}\bigl(\zb^\pm\bigr)$.
One verifies that
\[
\Delta(\ab^*):=\prod_{i<j\le N}\bigl(-a_{N-i}^{-1}+a_{N-j}^{-1}\bigr)
=(-1)^{-\frac 12 N(N-1)}\Delta(\ab)\prod_{i=1}^N a_i^{-(N-1)}
\]
and we obtain
\begin{align}
\Delta(\bb^*/{-}\ab^*):={}&
(-1)^{N^2}\frac{\Delta(\bb^*)\Delta(\ab^*)}{\prod_{i,j=1}^N\bigl(a_i^{-1}+b_j^{-1}\bigr)}\label{Delta(-b^-1,-a^-1)}\\
 ={}&
\det \bigl(\bigl(-a_{N-i}^{-1}-b_{N-j}^{-1}\bigr)^{-1}\bigr)_{i,j=1,\dots,N}
 =(-1)^{N^2}\Delta(\ab/{-}\bb)\prod_{i=1}^N a_i b_i\label{Delta(-b^-1,-a^-1)=}
\end{align}
and
\begin{gather*}
\Delta^{\rm B}\bigl(\zb^{\pm*}\bigr)=\Delta^{\rm B}\bigl(\zb^\pm\bigr)=\Delta^{\rm B}\bigl(\hat{\zb}^{\pm*}\bigr)=\Delta^{\rm B}\bigl(\hat{\zb}^\pm\bigr).
\end{gather*}

{\bf Polarization $\boldsymbol{\PP(\ab,\bb)}$ of the coordinate set $\boldsymbol{\ab}$, $\boldsymbol{\bb}$.}

\begin{de}
A {\em polarization} of the pair of the ordered sets $(\ab,\bb) = (a_1,\dots,a_N,b_1,\dots, b_N)$,
is a pair of ordered sets $\zb=\bigl(\zb^+, \zb^-\bigr)$
 with cardinalities (or {\em lengths})
\begin{gather*}
N(\zb^+):= \#(\zb^+), \qquad N(\zb^-):= \#(\zb^-),
\end{gather*}
 satisfying
\begin{gather*}
\zb^+ \cap \zb^- = \ab \cap \bb, \qquad \zb^+ \cup \zb^- = \ab \cup \bb .
\end{gather*}
The set of all polarizations of $\ab$, $\bb$ is denoted by $\PP(\ab,\bb)$.
\end{de}

Introduce the following notations:
$\tilde{Q} := \ab \cap \bb$, $\tilde{C}:=\ab \cup \bb$
and their cardinalities as
$
\tilde{q} := \#\bigl(\tilde{Q}\bigr)$, $\tilde{c} := \#\bigl(\tilde{C}\bigr)=2N-\tilde{q}$.
We have
$N(\zb^+)+ N(\zb^-)=2N$,
so $N\bigl(\zb^\pm\bigr)$ must have the same parity. It is easy to see that the cardinality of
$\PP(\ab, \bb)$ is $2^{2N-2\tilde{q}}$.

\begin{lem}[binary sequence associated to a polarization]
For every polarization $\zb:= (\zb^+, \zb^-)$ of $(\ab, \bb)$ there is a unique binary sequence
of length $2N$
\begin{gather*}
\epsilon(\zb) =(\epsilon_1(\zb), \dots, \epsilon_{2N}(\zb)),
\end{gather*}
with
$\epsilon_j(\zb) =\pm$, $j=1, \dots 2N$,
such that the sequence of pairs
\begin{gather}
((a_1, \epsilon_1(\zb)), \dots (a_N, \epsilon_N(\zb)), (b_1, \epsilon_{N+1}(\zb)), \dots , ( b_N, \epsilon_{2N}(\zb))
\label{a_b_sequence}
\end{gather}
is a permutation of the sequence
\be
\bigl(\bigl(z^+_1,+\bigr), \dots \bigl( z^+_{N(\zb^+)},+\bigr), \bigl(z^-_1, -\bigr), \dots , \bigl(z^-_{N(\zb^-)}, -\bigr)\bigr) .
\label{z_sequence}
\ee
\end{lem}

\begin{de}
 The {\em sign} of the polarization $(\zb^+, \zb^-)$, denoted by $\sgn(\zb)$, is defined
as the sign of the permutation that takes the sequence (\ref{a_b_sequence}) into the sequence
(\ref{z_sequence}).
\end{de}
Denote by
\begin{gather*}
\pi\bigl(\zb^\pm\bigr) :=\#\bigl(\ab \cap \zb^{\pm}\bigr)
\end{gather*}
the cardinality of the intersection of $\ab$ with $\zb^{\pm}$. It follows that
\begin{gather*}
\pi(\zb^+) + \pi(\zb^-) = N +\tilde{q}.
\end{gather*}

We introduce the following notations:
\begin{gather}
\label{d*}
 \left[\zb^+,\zb^-\atop \ab,\bb\right] :=\frac{\Delta^{\rm B}(\zb^+)\Delta^{\rm B}(\zb^-)}{\Delta(\bb^*/{-}\ab^*)}
{(-1)^{\frac{1}{2}N(N+1)+q}\over 2^{N-q}}
\sgn(\zb)(-1)^{\pi(\zb^-)+\frac12 m(\hat{\zb}^-)},\\
\label{d}
\left[{\zb}^+,{\zb}^-\atop {\ab},{\bb}\right]^*
:= (-1)^{{N}^2}
\frac{\Delta^{\rm B}\bigl({\zb}^+\bigr)\Delta^{\rm B}({\zb}^-)}{\Delta({\ab}/{-}{\bb})}
{(-1)^{\frac{1}{2}{N}({N}+1)+{q}}\over 2^{{N}-{q}}} \nonumber\\
\hphantom{\left[{\zb}^+,{\zb}^-\atop {\ab},{\bb}\right]^*:=}{} \times
\sgn({\zb})(-1)^{\pi({\zb}^-)+\frac12 m(\hat{{\zb}}^-)}
\prod_{i=1}^{{N}} {a}_i{b}_i,
\end{gather}
where $\Delta(\bb^*/{-}\ab^*)$ is given by (\ref{Delta(-b^-1,-a^-1)}).

\subsection[Charged and neutral fermions and currents {[12]}]{Charged and neutral fermions and currents \cite{JM}}\label{free_fermions}
The fermionic creation and annihilation operators satisfy the anticommutation relations
\be
[\psi_j,\psi_k]_+= \bigl[\psi^\dag_j,\psi^\dag_k\bigr]_+=0,\qquad \bigl[\psi_j,\psi^\dag_k\bigr]_+=\delta_{jk} .
\label{charged-canonical}
\ee

The {\em vacuum} element $|n\rangle$ in each charge sector
$\FF_n$ is the basis element
corresponding to the trivial partition $\lambda = \varnothing$:
\[
| n\rangle :=|\varnothing; n \rangle = e_{n-1} \wedge e_{n-2} \wedge \cdots .
\]
Elements of the dual space $\FF^*$ are denoted by {\em bra} vectors $\langle w |$,
with the dual basis $\{\langle \lambda ;n|\}$ for~$\FF^*_n$ defined by the pairing
$\langle \lambda; n | \mu; m\rangle = \delta_{\lambda \mu} \delta_{nm}$.
For KP $\tau$-functions, we need only consider the $n=0$ charge sector $\FF_0$,
and generally drop the charge $n$ symbol, denoting the basis elements simply as
$|\lambda\rangle :=|\lambda;0\rangle$.
 For $j>0$, $\psi_{-j}$ and $\psi^\dag_{j-1}$
\big(resp.\ $\psi^\dag_{-j}$ and $\psi_{j-1}$\big) annihilate the right (resp.\ left) vacua:
\begin{gather}
\psi_{-j} |0\rangle =0, \qquad \psi^\dag_{j-1} |0\rangle =0, \qquad \forall j >0,
\label{vac_annihil_psi_j_r}
 \\
\langle 0| \psi^\dag_{-j} =0, \qquad \langle 0 | \psi_{j-1} =0, \qquad \forall j >0.
\nonumber
\end{gather}

Neutral fermions $\phi^+_j$ and $\phi^-_j$ are defined \cite{DJKM1} by
 \be
\phi^+_j :=\frac{\psi_j+ (-1)^j\psi^\dag_{-j}}{\sqrt 2},\qquad
\phi^-_j := {\rm i}\frac{\psi_j-(-1)^j \psi^\dag_{-j}}{\sqrt 2},\qquad j \in \mathbb{Z}
\label{charged-neutral}
\ee
(where ${\rm i}=\sqrt{-1}$), and satisfy
\be
\label{neutral-canonical}
 \bigl[\phi^+_j,\phi^-_k\bigr]_+=0,\qquad \bigl[\phi^+_j,\phi^+_k\bigr]_+ = \bigl[\phi^-_j,\phi^-_k\bigr]_+ =(-1)^j \delta_{j+k,0}.
\ee
In particular,
\begin{gather*}
\bigl(\phi^+_0\bigr)^2=(\phi^-_0)^2=\frac{1}{2}.
\end{gather*}
Acting on the vacua $|0\rangle$ and $|1 \rangle$, we have
\begin{gather}
\phi^+_{-j}|0\rangle \&= \phi^-_{-j} |0\rangle = \phi^+_{-j}|1\rangle =
\phi^-_{-j} |1\rangle =0 , \qquad \forall j > 0 , \quad \forall j > 0,
\label{phi_vac_r}
\\
\langle 0| \phi^+_{j} \&= \langle 0|\phi^-_{j} = \langle 1| \phi^+_{j} =
\langle 1|\phi^-_{j} =0 , \qquad \forall j > 0 ,
\nonumber
\\
\phi^+_0|0\rangle \& =
- {\rm i} \phi^-_0 |0\rangle =
\frac{1}{\sqrt{2}} \psi_0|0\rangle = \frac{1}{\sqrt{2}} |1\rangle ,
\label{phi_0_vac_r}
\\
 \langle 0| \phi^+_0 \& =
 {\rm i} \langle 0|\phi^-_0 =
\frac{1}{\sqrt{2}} \langle 0 | \psi_0^\dag = \frac{1}{\sqrt{2}}\langle 1|.
\nonumber
\end{gather}

{\bf Lie groups. Factorization condition.}
Let us define the normal ordering ${:}\psi_j\psi_k{:}$ of $\psi_j\psi_k$ as
$\psi_j\psi_k -\langle 0|\psi_j\psi_k|0\rangle $
and the normal ordering of $\phi^\pm_j\phi^\pm_k$ as
${:}\phi^\pm_j\phi^\pm_k{:} =\phi^\pm_j\phi^\pm_k -\langle 0|\phi^\pm_j\phi^\pm_k|0\rangle $.

Let us denote
\be\label{hat_g_via_fermions}
\hat{g}={\rm e}^{\sum_{j,k} A_{j,k}{:}\psi_j\psi^\dag_k{:}},
\ee
where $A_{j,k}$ are complex numbers.

If we ask $A_{j,k}$ decay in a fast enough way as $|j-k|\to\infty$,
then the exponents form Lie algebras~$\widehat{\mathfrak{gl}}_\infty$, and
 elements (\ref{hat_g_via_fermions}) form ${\rm GL}_\infty$ group.

\begin{rem}
``Fast enough'' implies the possibility of the exponentials to form Lie algebra, where the
commutator is well defined. Say, finite matrices satisfy this. The other example are the generalized
Jacobian matrices (matrices with a finite number of nonvanishing diagonals).
\end{rem}

The elements
\be\label{h(B)}
\hat{h}^\pm(B)={\rm e}^{\sum_{j,k} B_{j,k}{:}\phi^\pm_j\phi^\pm_k{:}},
\ee
where $B_{j,k}$ are complex numbers and $B_{j,k}=-B_{k,j}$.

If we ask $B_{j,k}$ decay in a fast enough way as $|j-k|\to\infty$,
the elements $\hat{h}^\pm(B)$ form $B_\infty$ group.
One can check the equality
\begin{gather*}
(-1)^k{:}\psi_j\psi^\dag_{-k}{:} - (-1)^j {:}\psi_k \psi^\dag_{-j}{:} = {:}\phi^+_j\phi^+_k{:} + {:}\phi^-_j \phi^-_k{:}.
\end{gather*}
Then we have the important theorem.

\begin{lem}[\cite{JM}]\label{g=h^+h^-Lemma} Suppose
\be\label{g_B_infty}
\hat{g}= {\rm e}^{\sum_{j,k} B_{jk} ( (-1)^k{:}\psi_j\psi^\dag_{-k}{:} - (-1)^j {:}\psi_k \psi^\dag_{-j}{:} )},
\ee
where $B$ is a given $($perhaps infinite$)$ antisymmetric matrix.
Then
\begin{gather*}
 \hat{g} = \hat{h}^+\hat{h}^-,
\end{gather*}
where $\hat{h}^\pm$ are defined by \eqref{h(B)}
\[
\hat{h}^\pm(B)={\rm e}^{\sum_{j,k} B_{j,k}{:}\phi^\pm_j\phi^\pm_k{:}}.
\]
\end{lem}

\begin{rem}
The form (\ref{g_B_infty}) says that $\hat{g}\in B_\infty$.
\end{rem}

{\bf Fermi fields.} Introduce
\[
\psi(z)=\sum_{i\in\mathbb{Z}} z^i\psi_i,\qquad \psi^\dag(z)=\sum_{i\in\mathbb{Z}} z^{-i-1}\psi^\dag_i
\qquad \text{and}\qquad
\phi^\pm(z)=\sum_{i\in\mathbb{Z}} z^i\phi^\pm_i ,
\]
where $z$ is a formal parameter.\footnote{In Euclidian 2D QFT $z={\rm e}^{{\rm i}\varphi-\tau}$, where $\varphi$ has
the meaning of the coordinate of the fermion $\psi(z)$ (of $\phi(z)$) and~$\tau$ has the meaning of time variable. }
From (\ref{charged-neutral}), it follows
\be\label{psi(z)-phi(z)}
\psi(z) = \frac { \phi^+(z) - {\rm i} \phi^-(z)}{\sqrt{2}}, \qquad
 \psi^\dag(-z)= \frac{\phi^+(z) + {\rm i} \phi^-(z)}{z\sqrt{2}} .
\ee

\begin{rem}
The formula
\[
\psi_j =\frac{\phi^+_j - {\rm i}\phi^- _{j}}{\sqrt 2},\qquad
 (-1)^j \psi^\dag_{-j} =\frac{\phi^+_j+ {\rm i} \phi^-_{j}}{\sqrt 2}
\]
is quite similar to formula (\ref{psi(z)-phi(z)}) and is rather similar to
the relations
\[
\psi\bigl(-z^{-1}\bigr) =\frac{\phi^+\bigl(-z^{-1}\bigr) - {\rm i}\phi^-\bigl(-z^{-1}\bigr)}{\sqrt 2},\qquad
\psi^\dag\bigl(z^{-1}\bigr) =- z\frac{\phi^+\bigl(-z^{-1}\bigr)+ {\rm i} \phi^-\bigl(-z^{-1}\bigr)}{\sqrt 2}.
\]
These equations result in the useful equalities
\begin{gather} 
(-1)^j\psi_j\psi^\dag_{-j}= {\rm i}\phi^+_j\phi^-_j +\frac12 \delta_{j,0} , \nonumber
\\ \label{psipis=phiphi-fields}
 \psi(z) \psi^\dag(-z)=\frac{{\rm i}}{z}\phi^+(z)\phi^-(z) ,
\\ \label{psipis=phiphi-fields'}
 \psi\bigl(-z^{-1}\bigr) \psi^\dag\bigl(z^{-1}\bigr)=-{\rm i}z\phi^+\bigl(-z^{-1}\bigr)\phi^-\bigl(-z^{-1}\bigr) .
\end{gather}
\end{rem}

We say that $\psi_j$, $\psi^\dag_j$ and $\phi^\pm_j$ are the Fermi modes of the Fermi fields
$\psi(z)$, $\psi^\dag_j(z)$ and $\phi^\pm(z)$, respectively.

{\bf Pairwise VEV.}
From
\begin{gather*}
\langle 0|\psi_k\psi^\dag_j |0\rangle = \begin{cases}\delta_{k,j},& k<0,\\
 0,& k>0,
 \end{cases}
\\
\langle 0|\phi^\pm_k\phi^\pm_j|0\rangle =\begin{cases}(-1)^k\delta_{k,-j},& k<0,\\
 0,& k>0 ,
 \end{cases}
\\
\langle 0|\phi^\pm_k\phi^\mp_j|0\rangle =\pm \frac {\rm i}2 \delta_{k,0}\delta_{j,0},
\end{gather*}
one can easily obtain by direct calculation
\begin{gather}\label{pairwise-corr}
\langle 0|\psi(a_k)\psi^\dag(-b_j) |0\rangle = \frac{1}{b_j+a_k} ,
\\ \label{VEV-phi-phi}
\langle 0|\phi^\pm(a_k)\phi^\pm(a_j)|0\rangle =\frac 12 \frac{a_k-a_j}{a_k+a_j} ,
\\
\langle 0|\phi^\pm(a_i)\phi^\mp(a_j)|0\rangle =\pm \frac {\rm i}2 . \nonumber
\end{gather}

\begin{rem}
Actually
the time order\footnote{In 2D QFT the argument of
the Fermi field is written as $x={\rm e}^{\sqrt{-1}\varphi - \tau}$, where $\varphi$ is the space and $\tau$ is the time
coordinate of the fermion.}
(the operation $\mathbb{T}$) is implied for the pairwise correlation function (say, in case of neutral
fermions)
\[
\langle 0|\mathbb{T}\bigl[\phi^\pm(z_a)\phi^\pm(z_b)\bigr]|0\rangle :=\begin{cases}
\langle 0|\phi^\pm(z_a)\phi^\pm(z_b)|0\rangle & {\rm if}\ |z_a| > |z_b|,\\
-\langle 0|\phi^\pm(z_b)\phi^\pm(z_a)|0\rangle & {\rm if}\ |z_a| < |z_b| .
\end{cases}
\]
Say, for $|z_a|>|z_b|$, we have
\begin{gather*}
\langle 0|\phi^\pm(z_a)\phi^\pm(z_b)|0\rangle=
\langle 0|\sum_{i}\phi^\pm_i (z_a)^i\sum_j\phi^\pm_j (z_b)^j|0\rangle=\frac 12+\sum_{j > 0} (-1)^{j}\biggl(\frac{z_b}{z_a}\biggr)^j
=\frac 12 \frac{1-\frac{z_b}{z_a}}{1+\frac{z_b}{z_a}} .
\end{gather*}
While in case $|z_a|<|z_b|$ we get
\[
-\frac 12\frac{1-\frac{z_a}{z_b}}{1+\frac{z_a}{z_b}}.
\]
 In
{\em both} cases ($|z_a| > |z_b|$ and $|z_a| < |z_b|$), the answer
can be written as $\frac{z_a - z_b}{z_a+z_b}$ (and the limit $|z_a|\to |z_b|$ where $z_a\neq -z_b$ is smooth).
We see, if our goal is to present $\frac 12 \frac{z_a - z_b}{z_a+z_b}$
as the correlation function, we imply the time ordering. The same convention will be true for higher
correlation functions of neutral fermions. The point is that the final formulas are clever enough not to take
special care about the time ordering.

Next, the time ordering for the charged fermions has the form
\[
\langle 0|\mathbb{T}'\bigl[\psi(x)\psi^\dag(y)\bigr]|0\rangle :=\begin{cases}
\langle 0|\psi(x)\psi^\dag(y)|0\rangle & {\rm if}\ |x| > |y| ,\\
-\langle 0|\psi^\dag(y)\psi(x)|0\rangle & {\rm if}\ |x| < |y| .
\end{cases}
\]
In case $|x|=|y|$ and $x\neq y$, the fermions $\psi(x)$ and $\psi^\dag(y)$ anticommute
\big(Dirac delta function on the circle, $\frac{dy}{y}\sum_{n\in\mathbb{Z}}\frac{a^n}{b^{n}}$, vanishes\big).

There is the following matching:
\[
\langle 0|\mathbb{T}'\bigl[\psi(x)\psi^\dag(-y)\bigr]|0\rangle
=\langle 0|\mathbb{T}\biggl[\biggl(\frac { \phi^+(x) - {\rm i} \phi^-(x)}{\sqrt{2}} \biggr)
\biggl(\frac{\phi^+(y) + {\rm i} \phi^-(y)}{y\sqrt{2}} \biggr)\biggr]|0\rangle .
\]
So $\mathbb{T}'$-ordering is consistent with $\mathbb{T}$-ordering. This is also correct
for all higher correlators.

In what follows, we shall not return to this topic and will never use $\mathbb{T}$ symbols.
\end{rem}

{\bf Factorization lemma.}
As it was done in \cite{HO1,HO2,HO3}, we need Lemmas \ref{Polarization-fields}--\ref{Psi-lambda-Phi-mu-Phi-mu} to calculate tau functions we have

 \begin{lem}[factorization] \label{factorization_lemma}
 If $U^+$ and $U^-$ are either even or odd degree elements of the subalgebra generated by the operators
 $\bigl\{\phi^+_i\bigr\}_{i \in \Zb}$ and $\{\phi^-_i\}_{i \in \Zb}$, respectively, the VEV of their product can be factorized as
\begin{gather*}
\langle 0 | U^+ U^-|0\rangle =
\begin{cases}
\langle 0 | U^+|0\rangle \langle 0| U^- |0\rangle
&
\text{if } U^+ \text{ and } U^- \text{ are both of even degree},\\
0 & \text{if } U^+ \text{ and } U^- \text{ have different parity},\\
2{\rm i} \langle 0|U^+\phi^+_0|0\rangle \langle 0|U^-\phi^-_0 |0\rangle &
\text{if } U^+ \text{ and } U^+ \text{ are both of odd degree}.
\end{cases}
\end{gather*}
\end{lem}

{\bf Currents.}
Define currents
\begin{gather*}
J_m=\sum_{i\in Z} \psi_i\psi^\dag_{i+m},\quad m=\pm 1,\pm 2,\pm 3,\dots ,
\\ 
 J^{{\rm B}\pm}_m=\frac 12 \sum_{i\in Z} (-)^i\phi^\pm_{-i-m}\phi^\pm_{i},\quad m\ \text{odd} .
\end{gather*}
We have from (\ref{charged-neutral})
\begin{gather}\label{J=J^B+J^B}
J_m = J^{{\rm B}+}_m + J^{{\rm B}-}_m,\quad m\ \text{odd} ,
\\
J_m = \sqrt{-1}\sum_{j\in\mathbb{Z}} \phi^+_{j-m}\phi^-_{-j},\quad m\ \text{even} . \nonumber
\end{gather}
The currents form the Heisenberg algebras as follows:
\begin{gather}\label{Hei}
[J_m,J_n]=m\delta_{m,n},\qquad m\neq 0 ,
\\ \label{Hei-B}
\bigl[J^{{\rm B}\pm}_m,J^{{\rm B}\pm}_n\bigr] = \frac 12 m\delta_{m+n,0},\qquad
\bigl[J^{{\rm B}+}_m,J^{{\rm B}-}_n\bigr]=0,\qquad m,n\ \text{odd}.
\end{gather}
One can see that
\begin{gather}\label{J-vac}
J_m|0\rangle = 0 =\langle 0|J_{-m},\qquad m>0.
\\
\label{J^B-vac}
J^{{\rm B}\pm}_m|0\rangle =0= \langle 0|J^{{\rm B}\pm}_{-m},\qquad m>0.
\end{gather}

{\bf Partitions for products of Fermi modes.}
For a given $\lambda=(\alpha|\beta)\in\Pa$, let us use the following notations
\begin{gather}\label{Psi-lambda}
\Psi_\lambda
 :=(-1)^{\sum_{j=1}^r\beta_j}(-1)^{\frac{1}{2}r(r-1)}
\psi_{\alpha_1}\cdots
\psi_{\alpha_r}\psi^\dag_{-\beta_1-1} \cdots \psi^\dag_{-\beta_r-1} ,\\
\label{Psi^dag-lambda}
\Psi^*_{\lambda} :=(-1)^{\sum_{j=1}^r\beta_j}(-1)^{\frac{1}{2}r(r-1)}
\psi_{-\beta_r-1}\cdots \psi_{-\beta_1-1} \psi^\dag_{\alpha_r} \cdots
\psi^\dag_{\alpha_1} .
\end{gather}
 Note that $|\lambda\rangle = \Psi_\lambda|0 \rangle$,
$\langle \lambda|=\langle 0|\Psi^*_\lambda$.

Let $\alpha=(\alpha_1,\dots,\alpha_k)\in\DP$, where $\alpha_k\ge 0$. (Thus, you pay attention
on the last part of $\alpha$ which can be equal to $0$.)
We also use the notation $m(\alpha)$ for the number $k$. (Notice that $m(\alpha)$ is
either $\ell(\alpha)$ (which is the number of non-zero parts of $\alpha$), or it is $\ell(\alpha)+1$).

Introduce
\begin{gather*}
\Phi^\pm_\alpha := 2^{\frac k2}\phi^\pm_{\alpha_1}\cdots \phi^\pm_{\alpha_{k}} ,\qquad
\Phi^\pm_{-\alpha} := (-1)^{\sum_{i=1}^k\alpha_i}2^{\frac k2}\phi^\pm_{-\alpha_k}\cdots \phi^\pm_{-\alpha_{1}} ,
\end{gather*}
where $k=m(\alpha)$.
Apart from the products $\Phi^\pm_\alpha$ the products $\Phi^\pm_{\hat{\alpha}}$ will be of
use where $\hat{\alpha}$ denotes the supplemented partition, see Definition \ref{supplemented-partition}.
Then we get
\begin{gather}\label{orthonormalitb_Psi}
\langle 0|\Phi^\pm_{-\alpha}\Phi^\pm_\beta |0\rangle =
\langle 0|\Phi^\pm_{-\hat{\alpha}}\Phi^\pm_{\hat{\beta}} |0\rangle =2^{\ell(\alpha)}\delta_{\alpha,\beta} ,
\\ \label{orthogonalitb_Phi}
\langle 0|\Phi^\pm_{-\alpha}\Phi^\mp_\beta |0\rangle =
\langle 0|\Phi^\pm_{-\hat{\alpha}}\Phi^\mp_{\hat{\beta}} |0\rangle =\pm\frac {\rm i}2 \delta_{\alpha,0}\delta_{\beta,0} .
\end{gather}

 In what follows, sometimes we shall write $\phi_{x}$ and $\Phi_{x}$ instead of $\phi^+_{x}$ and $\Phi^+_{x}$
 omitting superscripts $\pm$. We hope it will not produce a confusion.

{\bf Notations for products of Fermi fields.}
Next, we introduce
\begin{gather}\label{Psi-x}
\Psi(\ab/{-}\bb):=\prod_{i=1}^N \psi(a_i)\psi^\dag(-b_i)\\ \hphantom{\Psi(\ab/{-}\bb)}{}
\label{Psi-x1}
=(-1)^{\frac12 N(N-1)} \psi(a_1)\cdots \psi(a_N)\psi^\dag(-b_1) \cdots \psi^\dag(-b_N),
\\
\Psi(\bb^*/{-}\ab^*)= \prod_{i=1}^N \psi\bigl(-b_i^{-1}\bigr)\psi^\dag\bigl(a_i^{-1}\bigr) \nonumber\\ \hphantom{\Psi(\bb^*/{-}\ab^*)}{}
=(-1)^{\frac12 N(N-1)} \psi\bigl(-b_N^{-1}\bigr)\cdots \psi\bigl(-b_1^{-1}\bigr) \psi^\dag\bigl(a_N^{-1}\bigr) \cdots \psi^\dag\bigl(a_1^{-1}\bigr) \nonumber
\end{gather}
and
\begin{gather*}
\Phi^{\pm}(\zb^\pm):=
2^{\frac {N(\zb^\pm)}{2} }\phi^\pm\bigl(z^\pm_1\bigr)\cdots \phi^\pm\bigl(z^\pm_{N(\zb^\pm)}\bigr),\\
\Phi^{\pm}\bigl(\zb^{\pm*}\bigr)= 2^{\frac {N(\zb^\pm)}{2} }
\phi^\pm\bigl(z^{\pm*}_1\bigr)\cdots \phi^\pm\bigl(z^{\pm*}_{N(z^\pm)}\bigr),
\end{gather*}
where each $z^{\pm*}_i$ is defined in (\ref{x*y*}).

Now consider VEV of these products. By (\ref{pairwise-corr}) and by the Wick's rule (see Appendix \ref{wick_app}),
we have for (\ref{Psi-x})
\begin{gather}\label{VEV-Psi}
\langle 0|\Psi(\ab/{-}\bb)|0\rangle =
\det \bigl( \langle 0|\psi(a_i)\psi^\dag(-b_j) |0\rangle\bigr)_{i,j=1,\dots,N}
=\Delta(\ab/{-}\bb). 
\end{gather}

Then thanks to (\ref{Delta(-b^-1,-a^-1)=}), we get
\begin{gather*}
\langle 0|\Psi(\bb^*/{-}\ab^*)|0\rangle = (-1)^{N^2}\Delta(\ab/{-}\bb)\prod_{i=1}^N a_i b_i .
\end{gather*}

Similarly, from (\ref{VEV-phi-phi})
by Wick's rule, we obtain
 \begin{gather}\label{VEV-Phi}
\langle 0|\Phi^\pm(\hat{\zb})|0\rangle =
\Pf \bigl( \langle 0|\phi^\pm(z_i)\phi^\pm(z_j) |0\rangle\bigr)_{i,j=1,\dots,N(\hat{\zb})}
=\Delta^{\rm B}(\zb) .
\end{gather}
Then it follows that
 \begin{gather*}
\langle 0|\Phi^\pm(\hat{\zb}^*)|0\rangle =
\Delta^{\rm B}(\zb^*)
= \Delta^{\rm B}(\zb) .
\end{gather*}

\begin{rem}
One can make sure that
\begin{align}
\langle 0|\Psi(\ab/{-}\ab)|0\rangle=\Delta(\ab/{-}\ab)&{}=
\bigl(\Delta^{\rm B}(\ab)\bigr)^2 \prod_{j=1}^N (2a_j)^{-N} \nonumber
\\&{}=
\langle 0|\Phi^+(\ab)\Phi^-(\ab)|0\rangle \prod_{j=1}^N (2a_j)^{-N} .\label{VEV-Psi=VEV-PhiPhi}
\end{align}

The equality (\ref{VEV-Psi=VEV-PhiPhi}) can be obtained in a different way as follows:
\begin{gather*}
\psi(a_j)\psi^\dag(-a_j)=\frac{{\rm i}}{a_j}\phi^+(a_j)\phi^-(a_j)
\end{gather*}
one gets
\[
\Psi(\ab/{-}\ab)= (-1)^{\frac 12 N(N-1)} {\rm i}^N \Phi^+(\ab)\Phi^-(\ab)\prod_{j=1}^N \frac{1}{a_j} .
\]
Then we obtain (\ref{VEV-Psi=VEV-PhiPhi})
from (\ref{Psi-x}), (\ref{VEV-Psi}) and (\ref{pairwise-corr}), (\ref{VEV-phi-phi}), (\ref{VEV-Phi}).
\end{rem}

Let us introduce
\begin{gather*}
\tilde{\Psi}(\ab/{-}\bb) = \frac{\Psi(\ab/{-}\bb)}{\langle 0|\Psi(\ab/{-}\bb)|0\rangle}
\qquad \text{and}\qquad
\tilde{\Phi}^{\pm}(\hat{\zb})=\frac{\Phi^{\pm}(\hat{\zb})}{\langle 0|\Phi^\pm(\hat{\zb})|0\rangle}.
\end{gather*}
So we have
\be\label{normalization}
\langle 0|\tilde{\Psi}(\ab/{-}\bb)|0\rangle =
\langle 0|\tilde{\Phi}(\hat{\zb})|0\rangle =1.
\ee

In what follows apart of the products $\Phi^\pm(\zb)$, the products $\Phi^\pm(\hat{\zb})$ will be of
use where $\hat{\zb}$ denotes the supplemented coordinate set, see Definition \ref{supplemented-coordinate-sets}.

{\bf Partitions for currents.}
We introduce
\begin{gather}\label{J-Delta}
J_\Delta := \prod_{i=1}^{\ell(\Delta)} J_{\Delta_i},\qquad
J_{-\Delta} := \prod_{i=1}^{\ell(\Delta)} J_{-\Delta_i},\qquad \Delta\in\Pa,\\
J^{{\rm B}\pm}_\Delta := \prod_{i=1}^{\ell(\Delta)} J^{{\rm B}\pm}_{\Delta_i},\qquad
J^{{\rm B}\pm}_{-\Delta} := \prod_{i=1}^{\ell(\Delta)} J^{{\rm B}\pm}_{-\Delta_i},\qquad \Delta\in\OP .\nonumber
\end{gather}

\subsection{Bosonization formulas}

The nice part of the classical integrability worked out by Kyoto school is a number of bosonization formulas.
Following \cite{DJKM1,DJKM2}, we consider
\begin{gather}\label{gamma}
\hat{\gamma}(\pb)={\rm e}^{\sum_{m>0} \frac 1m p_{m}J_{m}},\qquad
\hat{\gamma}^\dag(\pb)={\rm e}^{\sum_{m>0} \frac 1m p_{m}J_{-m}},\\
\label{gamma-B}
\hat{\gamma}^{{\rm B}\pm}\bigl(2{\pb^{\rm B}}\bigr)={\rm e}^{\sum_{m>0,\text{odd}} \frac 2m p^{\rm B}_{m}J^{{\rm B}\pm}_{-m}},\qquad
\hat{\gamma}^{\dag{\rm B}\pm}\bigl(2{\pb^{\rm B}}\bigr)={\rm e}^{\sum_{m>0,\text{odd}} \frac 2m p^{\rm B}_{m}J^{{\rm B}\pm}_{-m}},
\end{gather}
where $\pb=(p_{1},p_{2},p_{3},\dots)$ and ${\pb^{\rm B}}=\bigl(p^{\rm B}_{1},p^{\rm B}_{3},p^{\rm B}_{5},\dots\bigr)$
a given sets of parameters.

For our purposes,
let us introduce the following notations:
\be\label{pb(xb/-yb)}
\pb(\ab/{-}\bb)=(p_{ 1}(\ab/{-}\bb),p_{ 2}(\ab/{-}\bb),p_{ 3}(\ab/{-}\bb),\dots),
\ee
where
\[
p_{m}(\ab/{-}\bb)=\sum_{i=1}^{N}(a_i^{m} -(-b_i)^{m} ),
\]
and also
\be\label{tilde-pb(xb/-yb)}
{\pb^{\rm B}}(\zb)=\bigl(p^{\rm B}_{ 1}(\zb),p^{\rm B}_{ 3}(\zb),p^{\rm B}_{ 5}(\zb),\dots\bigr),
\qquad \text{where}
\quad
p^{\rm B}_{m}(\zb)=\sum_{i=1}^{N(\zb)} z_i^m .
\ee

\begin{rem}\label{when-p=p'} Let us note that for each $m$ and for all possible $\ab$, $\tilde{\ab}$, we get
\[
p_{2m}(\ab/{-}\ab)=0=\tilde{p}_{2m}(\tilde{\ab}/{-}\tilde{\ab}).
\]
\end{rem}

The following bosonization formulas are well known:
\begin{gather*}
\psi(a)\psi^\dag(b)= \frac{1}{a-b}
{\rm e}^{\sum_{m>0} \frac 1m (b^{-m}-a^{-m})J_{-m}}
{\rm e}^{\sum_{m>0} \frac 1m (a^{m}-b^{m})J_{m}},
\\
\phi(a)\phi(b)= \frac 12 \frac{a-b}{a+b}
{\rm e}^{-\sum_{m>0,\text{odd}} \frac 2m (b^{-m}+a^{-m})J^{\rm B}_{-m}}
 {\rm e}^{\sum_{m>0,\text{odd}} \frac 2m (a^{m}+b^{m})J^{\rm B}_{m}} ,
\end{gather*}
therefore one can write
\begin{gather*}
\hat{\gamma}^\dag(\pb(\bb^*/{-}\ab^*))\hat{\gamma}(\pb(\ab/{-}\bb))=\tilde{\Psi}^*(\ab/{-}\bb),\\
\hat{\gamma}^\dag(\pb(\ab/{-}\bb))\hat{\gamma}(\pb(\bb^*/{-}\ab^*))=\tilde{\Psi}(\ab/{-}\bb),
\\ 
 \hat{\gamma}^{\dag{\rm B}\pm}\bigl(2{\pb^{\rm B}}(\zb^*)\bigr) \hat{\gamma}^{{\rm B}_\pm}\bigl(2{\pb^{\rm B}}(\zb)\bigr)=\tilde{\Phi}^{*\pm}(\hat{\zb}),
 \qquad
 \hat{\gamma}^{\dag{\rm B}\pm}\bigl(2{\pb^{\rm B}}(\zb)\bigr) \hat{\gamma}^{{\rm B}_\pm}\bigl(2{\pb^{\rm B}}(\zb^*)\bigr)=\tilde{\Phi}^{\pm}(\hat{\zb}).
\end{gather*}

These relations are in agreement with (\ref{normalization}) because of (\ref{J-vac}) and (\ref{J^B-vac}).

From above and from (\ref{J-vac}) and (\ref{J^B-vac}), we get
\begin{gather}\label{bosonization}
\langle 0|\hat{\gamma}(\pb(\ab/{-}\bb))=\langle 0|\tilde{\Psi}(\bb^*/{-}\ab^*),\qquad
\hat{\gamma}^\dag(\pb(\ab/{-}\bb))|0\rangle =\tilde{\Psi}(\ab/{-}\bb) |0\rangle,
\\
\label{bosonization-B}
\langle 0| \hat{\gamma}^{{\rm B}_\pm}\bigl(2{\pb^{\rm B}}(\zb)\bigr) =\langle 0|\tilde{\Phi}^{\pm}(\hat{\zb}^*),\qquad
\hat{\gamma}^{\dag{\rm B}\pm}\bigl(2{\pb^{\rm B}}(\zb)\bigr)|0\rangle = \tilde{\Phi}^{\pm}(\hat{\zb})|0\rangle .
\end{gather}

\subsection{Fermions and symmetric functions}

The bosonization formulas are also manifested by the representation of the known symmetric polynomials
as fermionic vacuum expectation values. This is an interesting part of the soliton theory.

{\bf Power sums and Schur functions.}
For preliminary information about symmetric functions, we recommend the textbook \cite{Mac1}.
This is about polynomial functions symmetric in variables $a_1,\dots,a_N$, the widely known examples
are the so-called power sums (or, the same, Newton sums) labels by a multiindex
$\lambda=(\lambda_1,\dots,\lambda_n)\in\Pa$:
\begin{gather*}
\pb_\lambda(\ab):=p_{\lambda_1}p_{\lambda_2}\cdots p_{\lambda_n},\qquad {\rm where}\quad
p_m=\sum_{i=1}^N a_i^m
\end{gather*}
or the Schur functions labeled by $\lambda=(\lambda_1,\dots,\lambda_n)\in\Pa$:
\[
s_\lambda(\ab)=\frac{\det \bigl( a_i^{\lambda_j-j+N}\bigr)_{i,j\le N}}{\Delta(\ab)}
\]
(where it is supposed that $n \le N$). For the empty partition, we put $\pb_0(\ab)=0$ and $s_0(\ab)=1$.
The power sums and the Schur functions are related
\[
\pb_\lambda = \sum_{\mu\atop |\mu|=|\lambda|} \chi_\mu(\lambda) s_\mu,
\]
where the coefficients $ \chi_\mu(\lambda)$ are very important in many problems. $\chi_\mu(\lambda)$ has the
meaning of the character of the irreducible representation $\mu$ of the permutation group $S_d$, $d=|\mu|$
evaluated on the cycle class $\lambda$, see for instance the textbook \cite{Mac1}.

On the space of polynomial symmetric functions denoted by $\Lambda_N$ the scalar product is given~by%
\be\label{pp}
\langle\pb_\lambda,\pb_\mu\rangle = z_\lambda \delta_{\lambda,\mu}=
\langle 0|J_{-\lambda} J_\mu |0\rangle
,\qquad \lambda,\mu\in \Pa,
\ee
where $z_\lambda$ was defined in (\ref{z-lambda}). Notice that the relation does not
include any $N$-dependence. The Schur functions form the orthonormal basis there
\be\label{ss}
\langle s_\lambda,s_\mu \rangle=
\delta_{\lambda,\mu}
=\langle 0|\Psi^*_\lambda\Psi_\mu|0\rangle
,\qquad \lambda,\mu\in\Pa .
\ee
The last equalities in both formulas (\ref{pp}) and (\ref{ss}) is the result of the direct computation
(using respectively (\ref{Hei}), (\ref{J-vac}) and (\ref{charged-canonical}), (\ref{vac_annihil_psi_j_r})),
however it is not an occasion: it is the manifestation
of the boson-fermion correspondence which is very popular in 2D physics.

In power case, where the sum variables are labeled with only odd parts, namely if
$\lambda=(\lambda_1,\dots,\lambda_k)\in\OP$, we denote $\pb_\lambda$ by
\smash{$\pb^{\rm B}_\lambda=\prod_{i=1}^{\ell(\lambda)}p_{\lambda_i}^{\rm B}$}.
These symmetric functions form the subspace of $\Lambda_N$ denoted by $\Gamma$ in \cite[Part~III]{Mac1}.
There is the following natural scalar product on~$\Gamma$:%
\be\label{pBpB}
\langle\pb^{\rm B}_\lambda,\pb^{\rm B}_\mu\rangle_{\rm B} =
2^{-\ell(\lambda)} z_{\lambda}\delta_{\lambda,\mu}
=
\langle 0|J^{\rm B}_{-\lambda} J^{\rm B}_\mu |0\rangle
,\qquad \lambda,\mu\in \OP .
\ee
Also
\be\label{QQ}
\langle Q_\alpha,Q_\beta\rangle_{\rm B}=
2^{\ell{(\alpha)}}\delta_{\alpha,\beta}
=\langle 0|\Phi_{-\alpha}\Phi_{\beta}|0\rangle =\langle 0|\Phi_{-\hat{\alpha}}\Phi_{\hat{\beta}}|0\rangle
,\qquad \alpha,\beta \in\DP,
\ee
 see \cite{Mac1}. The validity of the second equalities in (\ref{pBpB}) and in (\ref{QQ}) is obtained
 by the direct evaluation with the help respectively of (\ref{Hei-B}), (\ref{J^B-vac}) and of
 (\ref{neutral-canonical})--(\ref{phi_0_vac_r}).

We also have
\begin{gather*}
s_\lambda\bigl({\bf J}^\dag\bigr)|0\rangle =\Psi_\lambda |0\rangle,\qquad \langle 0|s_\lambda({\bf J})=\langle 0|\Psi^*_\lambda,\qquad
\lambda=(\alpha|\beta)\in\Pa,
\\
Q_\mu\bigl({\bf J}^{\dag{\rm B}\pm}\bigr)|0\rangle = \Phi_{\hat\mu}|0\rangle,
\qquad \langle 0|Q_\mu\bigl({\bf J}^{{\rm B}\pm}\bigr) = \langle 0|
\Phi_{-\hat{\mu}},\qquad \mu\in\DP,
\end{gather*}
where both Schur functions are considered as polynomials in power sum variables and the role
of power sums play respectively ${\bf J}=(J_1,J_2,J_3,\dots)$, ${\bf J}^\dag=\bigl(J^\dag_1,J^\dag_2,J^\dag_3,\dots\bigr)$ and
${\bf J}^{{\rm B}\pm}=\bigl(J^{{\rm B}\pm}_1,J^{{\rm B}\pm}_3,J^{{\rm B}\pm}_5,\dots\bigr)$,
${\bf J}^{\dag{\rm B}\pm}=\bigl(J^{\dag{\rm B}\pm}_1,J^{\dag{\rm B}\pm}_3,J^{\dag{\rm B}\pm}_5,\dots\bigr)$
(compare to
\cite{TauMM}).

We get
\begin{align}
\langle 0|\hat{\gamma}(\pb)\hat{\gamma}^\dag(\tilde{\pb})|0\rangle ={\rm e}^{\sum_{m>0}\frac 1m p_m\tilde{p}_m} =
\sum_{\lambda\in\Pa}\frac{1}{z_\lambda}\pb_\lambda\tilde{\pb}_\lambda
 =\sum_{\lambda\in\Pa} s_\lambda(\pb)s_\lambda(\tilde{\pb}),\label{vacuum-2-KP'}
\end{align}
which can be derived either from (\ref{Hei}) or also from (\ref{pp}).

From (\ref{Hei-B}),
\begin{align}
\langle 0|\hat{\gamma}^{{\rm B}\pm}\bigl(2{\pb^{\rm B}}\bigr)\hat{\gamma}^{\dag{\rm B}\pm}\bigl(2\tilde{\pb}^{\rm B}\bigr)|0\rangle
& ={\rm e}^{\sum_{m>0,\text{odd}}\frac 2m p_m\tilde{p}_m}
\nonumber\\
\label{vacuum-2-BKP'}
& =
\sum_{\lambda\in\OP}2^{\ell(\lambda)}\frac{1}{z_\lambda}\pb^{\rm B}_\lambda\tilde{\pb}^{\rm B}_\lambda =\sum_{\mu\in\DP} 2^{-\ell(\mu)} Q_\mu\bigl({\pb^{\rm B}}\bigr)Q_\mu\bigl(\tilde{\pb}^{\rm B}\bigr) .
\end{align}

Let us write down the following equalities:
\be\label{coherent_states}
\langle 0|\hat{\gamma}(\pb) = \sum_{\lambda\in\Pa} s_\lambda(\pb)\langle 0|\Psi^\dag_{\lambda} ,
\qquad
\hat{\gamma}^\dag(\tilde{\pb})|0\rangle =
\sum_{\lambda\in\Pa}\Psi_\lambda|0\rangle s_\lambda(\tilde{\pb}),
\ee
which, thanks to (\ref{orthonormalitb_Psi}), is equivalent to the Sato formula (\ref{Schur}) below, and
\begin{gather*}
\langle 0|\hat{\gamma}^{{\rm B}}\bigl({\pb^{\rm B}}\bigr)=
\sum_{\mu\in\DP}2^{-\ell(\mu)} Q_\mu\bigl({\pb^{\rm B}}\bigr)\langle 0|\Phi_{-\mu}
,\qquad
\hat{\gamma}^{\dag{\rm B}}\bigl(\tilde{\pb}^{\rm B}\bigr)|0\rangle =
\sum_{\mu\in\DP}2^{-\ell(\mu)}\Phi_\mu|0\rangle Q_\mu\bigl(\tilde{\pb}^{\rm B}\bigr) .
\end{gather*}
which, thanks to (\ref{orthogonalitb_Phi}), is equivalent to the relation found in \cite{You}, see (\ref{You_Q}) below.

\subsection{Fermions and tau functions: KP and BKP cases}\label{Fermions_and}

There are different treatments of the notion of tau function.
Here we use the fermionic approach to tau functions.

{\bf $\boldsymbol{\tau}$ functions as vacuum expectation values (VEVs).}\footnote{We give one of possible
definitions of the tau function. Here we want to avoid the notions of the Plucker
and the Cartan coordinates on Sato Grassmannian and isotropic Grassmannian which are out of real use
in the present work.}
According to \cite{JM}, KP tau functions can be presented in form of the following vacuum expectation value (VEV)
\be\label{tau-KP}
\tau(\pb|\hat{g})=\langle 0|\hat{\gamma}(\pb)\hat{g} |0\rangle,
\ee
where $\hat{g}$ is an exponential of a bilinear in $\{\psi_i\}$ and $\big\{\psi^\dag_i\big\}$ expression
\be\label{g=exp}
\hat{g} = {\rm e}^{\sum_{i,j} A_{ij} \psi_i\psi^\dag_j},
\ee
where $A$ is an infinite matrix which is treated as
$A\in\gl(\infty)$ (see Section~\ref{free_fermions}) and where
$\hat{\gamma}(\pb)$ is given by~(\ref{gamma}).

\begin{rem}\label{g_entries_exist}
Throughout the text, we assume that the so-called matrix elements
\be\label{g_mu_lambda}
\hat{g}_{\mu,\lambda}:=\langle 0|\Psi^\dag_\mu \hat{g} \Psi_\lambda |0\rangle
\ee
do exist for each pair $\mu,\lambda\in\Pa$. This assumption allows to identify tau functions
with their Taylor series in power sums by substituting (\ref{coherent_states}) and (\ref{g_mu_lambda})
into (\ref{tau-KP}).
\end{rem}

The tau function depends on the choice of $\hat{g}$ (the same, on the choice of the matrix $A$)
and on the infinite set of parameters
$\pb=(p_1,p_2,p_3,\dots)$ (the power sum variables).

\begin{rem}\label{adjoint}
Tau function (\ref{tau-KP}) can also be written as
\[
\tau(\pb|\hat{g})= \langle 0|\hat{g}^\dag\hat{\gamma}^\dag(\pb) |0\rangle,
\qquad \text{where}\quad
\hat{g}^\dag = {\rm e}^{\sum_{i,j} A_{ij} \psi_j\psi^\dag_i}
\]
and $\hat{\gamma}^\dag(\pb)$ is given by (\ref{gamma}).
\end{rem}

The BKP tau function depends on the set of odd-labeled power sums ${\pb^{\rm B}}=(p_1,p_3,\dots)$ and
can be presented as
\be\label{tau-BKP}
\tau^{B}\bigl(2{\pb^{\rm B}}|\hat{h}^\pm\bigr)=\langle 0|\hat{\gamma}^{\rm B\pm}\bigl(2{\pb^{\rm B}}\bigr)\hat{h}^\pm |0\rangle,
\ee
where $\hat{h}^\pm$ is an exponential of a quadratic in $\big\{\phi^\pm_i\big\}$ expression
\be\label{h=exp}
 \hat{h}^\pm ={\rm e}^{\sum_{i,j} B_{ij}\phi^\pm_i \phi^\pm_j },
\ee
where ${B}^\pm$ is infinite antisymmetric matrices
and where $\hat{\gamma}^{\rm B\pm}$ is given by (\ref{gamma-B}).

\begin{rem}\label{h_entries_exist}
We assume that the matrix elements
\[
\hat{h}^\pm_{\mu,\nu}:=\langle 0|\Phi^\pm_{-\mu} \hat{h} \Phi_\nu |0\rangle
\]
do exist for each pair $\mu,\nu\in\DP$.
\end{rem}

\begin{rem}
 The set of $t_m=\frac 1m p_m$, $m=1,2,3,\dots $, is called the set of the {\em KP
 higher times} and
 the set of $t_m^{\rm B}=\frac 2m p_m^{\rm B}$, $m=1,3,5,\dots $ is called the set of the {\em BKP
 higher times}.
\end{rem}

\begin{rem}\label{adjointB}
Tau function (\ref{tau-BKP}) can be also written as
\[
\tau^{\rm B}({\pb^{\rm B}}|\hat{h}^\pm)=
\langle 0|\bigl(\hat{h}^\pm\bigr)^\dag\hat{\gamma}^{\dag{\rm B}\pm}\bigl({\pb^{\rm B}}\bigr) |0\rangle,
\qquad \text{where}\quad
\bigl(\hat{h}^\pm\bigr)^\dag = {\rm e}^{\sum_{i,j} B_{ij} \phi^\pm_j\phi^\pm_i}
\]
and $\hat{\gamma}^{\dag{\rm B}\pm}\bigl({\pb^{\rm B}}\bigr)$ is given by (\ref{gamma-B}).
\end{rem}

In the basic works of Kyoto school (for instance, in \cite{JM}), the following relation was proved.
\begin{Theorem}[\cite{DJKM2}]
Under conditions
\begin{gather*}
\hat{g}=\hat{h}^+\hat{h}^- \in B_\infty
\end{gather*}
$($as written down in Lemma $\ref{g=h^+h^-Lemma})$ and also under
\[
{\pb^{\rm B}}=\bigl(p^{\rm B}_1,p^{\rm B}_3,p^{\rm B}_5,\dots\bigr),
\qquad
\pb':=(p_1,0,p_3,0,p_5,0,\dots),
\]
which are related by \eqref{p_B}--\eqref{p_p^B},
the relation between KP tau function \eqref{tau-KP} and $($any the both$)$ BKP tau functions \eqref{tau-BKP} is as follows:
\[
\bigl( \tau^{\rm B}\bigl({\pb^{\rm B}}|\hat{h}^\pm\bigr) \bigr)^2=\tau\bigl(\pb'|\hat{h}^+\hat{h}^-\bigr).
\]
\end{Theorem}

{\bf The Schur and the projective Schur functions as tau functions -- a remark.}

\begin{rem}\label{g,h,Psi,Phi} For any given $\lambda=(\alpha|\beta)\in\Pa$ and $\alpha\in\DP$,
there exists such $\hat{g}=\hat{g}(\lambda)$ and such $\hat{h}^\pm=\hat{h}^\pm(\alpha)$ that
\begin{gather*}
\hat{g}(\lambda)|0\rangle =\Psi_\lambda|0\rangle ,
\qquad
\hat{h}^\pm(\alpha)|0\rangle =\Phi^\pm_{\hat{\alpha}} |0\rangle .
\end{gather*}
This is the basic fact of the Sato theory. For a given $\lambda=(\alpha|\beta)$ and $\alpha$, one can present both
$\hat{g}(\lambda)$ and $\hat{h}^\pm(\alpha)$ in the explicit way, see appendix.
\end{rem}

The wonderful observation by Sato and his school \cite{JM,Sa} is the fermionic formula for
the Schur polynomial
\be\label{Schur}
s_\lambda(\pb)=\langle 0|\hat{\gamma}(\pb)\Psi_\lambda|0\rangle =
\langle 0|\Psi^*_\lambda\hat{\gamma}^\dag(\pb)|0\rangle .
\ee
In the BKP case, the similar formula was found in \cite{You}
\be\label{You_Q}
Q_\alpha\bigl({\pb^{\rm B}}\bigr)=\langle 0|\hat{\gamma}^{{\rm B}\pm}\bigl(2{\pb^{\rm B}}\bigr)\Phi^\pm_{\hat{\alpha}}|0\rangle
= \langle 0|\Phi^\pm_{-\hat{\alpha}}\hat{\gamma}^{\dag B\pm}\bigl(2{\pb^{\rm B}}\bigr)|0\rangle .
\ee

{\bf Schur functions $\boldsymbol{s_\lambda(\ab/{-}\bb)}$ and $\boldsymbol{Q_\mu(\zb)}$ and products of Fermi fields $\boldsymbol{\Psi(\ab/{-}\bb)}$,
$\boldsymbol{\Phi(\zb)}$.}
With the help of (\ref{bosonization}) and (\ref{bosonization-B}),
we rewrite formulas (\ref{Schur}) and (\ref{You_Q}) for both Schur functions as follows:
\be\label{s-PsiPsi}
s_\lambda(\pb(\ab/{-}\bb))=\langle 0|\tilde{\Psi}(\bb^*/{-}\ab^*)\Psi_\lambda|0\rangle
=\langle 0|\Psi^*_\lambda\tilde{\Psi}(\ab/{-}\bb)|0\rangle =:s_\lambda(\ab/{-}\bb),
\ee
where $\lambda=(\alpha|\beta)$ and
\begin{gather*}
Q_\alpha({\pb^{\rm B}}(\zb))=\langle 0|\tilde{\Phi}(\hat{\zb}^*)\Phi_{\hat{\alpha}}|0\rangle
=\langle 0|\Phi_{-\hat{\alpha}}\tilde{\Phi}(\hat{\zb})|0\rangle =:Q_\alpha(\zb) .
\end{gather*}

\begin{rem}
By the limiting procedure, we obtain from (\ref{s-PsiPsi})
\begin{gather}\label{s-Ff}
s_\lambda(\ab)\Delta(\ab)=\langle -N|\psi^\dag\bigl(a_1^{-1}\bigr)\cdots \psi^\dag\bigl(a_N^{-1}\bigr)\Psi_\lambda |0\rangle =
\langle 0|\Psi^*_\lambda \psi(a_1)\cdots \psi(a_N) |N\rangle,
\end{gather}
where to get the first equality we send $b_1> \cdots >b_N \to\infty$ in the second member of~(\ref{s-PsiPsi})
and to get the second equality in (\ref{s-Ff})we send $ b_N < \cdots <b_1\to 0$ in the last member of~(\ref{s-PsiPsi}).
\end{rem}

We rewrite (\ref{s-PsiPsi}) as
\[
s_\lambda(\ab/{-}\bb)\Delta(\ab/{-}\bb)=
\langle 0|\Psi(\bb^*/{-}\ab^*)\Psi_\lambda |0\rangle =(-1)^{N^2}
\langle 0|\Psi^*_\lambda \Psi(\ab/{-}\bb) |0\rangle \prod_{i=1}^N a_i^{-1}b_i^{-1} .
\]

\subsection[Two-sided KP tau function, two-sided BKP tau function, lattice KP tau function, lattice BKP tau function]{Two-sided KP tau function, two-sided BKP tau function,\\ lattice KP tau function, lattice BKP tau function}\label{2-sided,lattice}

All objects which we need are tau functions introduced in works of Kyoto school.
To be more precise, in the terminology concerning tau functions which we use, let us give definitions.

{\bf Two-sided KP tau functions. Lattice KP tau functions.}

\begin{de}
We call
\be\label{2-KP}
\tau(\pb,\tilde{\pb}|\hat{g})=\langle 0|\,\hat{\gamma}(\pb) \,\hat{g}\, \hat{\gamma}^\dag(\tilde{\pb})\,|0\rangle
\ee
the two-sided KP tau function. Here $\pb=(p_{1},p_{2},p_{ 3},\dots)$
and $\tilde{\pb}=(\tilde{p}_{1},\tilde{p}_{2},\tilde{p}_{ 3},\dots)$
are parameters. We call $t_m=\frac 1m p_m$ and $\tilde{t}_m=\frac 1m \tilde{p}_m$ the two-sided KP higher times.
\end{de}
This notion of the KP higher times is the same as it was introduced in \cite{DJKM2,MJD} and is common.
Two-sided KP tau function is usual (one-sided) KP tau function (\ref{tau-KP}) with respect to the higher times $\tb=\pb$
and at the same time it is one-sided KP tau function with respect to the higher time variables
$\tilde{\tb}=\tilde{\pb}$, see Remark \ref{adjoint}.\footnote{Tau function (\ref{2-KP}) can be also considered as the Toda lattice tau function
$\tau_N\bigl(\tb,\tilde\tb\bigr)$ \cite{JM,Takasaki_Schur,Takebe_Schur}, where $N=0$.}

\begin{de}
We call
\be
S_\lambda(\pb |\hat{g}):= \langle 0 |\hat{\gamma}(\pb) \hat{g} \Psi_\lambda| 0 \rangle
\label{largeS}
\ee
 lattice KP tau function labeled with a partition $\lambda\in\Pa$.
\end{de}

We do not assume that $S_\lambda(\pb)$ is a polynomial function in the variables $\pb$.
For the polynomial case, see \cite{HO2}.

\begin{rem}
The tau function (\ref{largeS}) can be considered as a discrete version of the two-way tau function
(\ref{2-KP}), where the dependence on continuous variables $p_1,p_2,\dots$ is replaced
by discrete variables that are parts of partitions. The tau function (\ref{largeS}) solves
the same discrete equations (which play the role of Hirota equations) with respect to the Frobenius parts of $\lambda$ as the Schur function $s_\lambda$ (Plucker relations).
\end{rem}

With the assumption of Remark \ref{g_entries_exist}, one can write both tau function as Taylor series
in power sum variables as follows:
\begin{gather}\label{Takasaki}
\tau(\pb,\tilde{\pb}|\hat{g})=\sum_{\mu,\lambda\in\Pa} \hat{g}_{\mu,\lambda}s_\mu(\pb)s_\lambda(\tilde{\pb}),
\\ \label{S_Taylor}
S_\lambda(\pb |\hat{g})=\sum_{\mu\in\Pa} \hat{g}_{\mu,\lambda}s_\mu(\pb),
\end{gather}
which results from (\ref{coherent_states}).
The series (\ref{Takasaki}) was written down in \cite{Takasaki_Schur,Takebe_Schur}
in the context of the study of the Toda lattice hierarchy.

For a given $\hat{g}$, the relation between two-sided and lattice KP tau functions is as follows:
\be\label{2-KP_lattice-KP}
\tau(\pb,\tilde{\pb}|\hat{g})=\sum_{\lambda\in\Pa} S_\lambda(\pb|\hat{g}) s_\lambda(\tilde{\pb}).
\ee

In the case $\hat{g}=1$, the last equality yields (see (\ref{vacuum-2-KP'}))
\[
\tau(\pb,\tilde{\pb}|\hat{g}=1)=\sum_{\lambda\in\Pa} s_\lambda(\pb) s_\lambda(\tilde{\pb})=
{\rm e}^{\sum_{m>0}\frac 1m p_m\tilde{p}_m},
\]
which is known as the Cauchy--Littlewood identity and occurs to be the simplest example of the two-sided KP tau function.

\begin{rem}
We also have
\begin{gather*}\label{Takasaki-generalization}
\tau(\pb,\tilde{\pb}|\hat{g}_1\hat{g}\hat{g}_2)=\sum_{\mu,\lambda\in\Pa} \hat{g}_{\mu,\lambda}S_\mu(\pb|\hat{g}_1)S_\lambda(\tilde{\pb}|\hat{g}_2),
\end{gather*}
which is a generalization of the Takasaki series \cite{Takasaki_Schur}.
\end{rem}

\begin{rem}\label{Plucker}
In the Grassmannian approach to the tau functions \cite{Sa}, the coefficients $S_\lambda(\pb|\hat{g})$
in (\ref{2-KP_lattice-KP}) have the meaning of the Plucker coordinates $\pi_\lambda(\hat{g})$ for the
one-side KP tau function, see Remark \ref{adjoint} where one should insert $\tilde{\pb}$ instead of $\pb$.
The notation was used in \cite{HO2} instead of $S_\lambda(\pb|\hat{g})$, see Remark \ref{notational_remark}.
\end{rem}

{\bf Two-sided BKP tau functions. Lattice BKP tau functions.}
\begin{de}
We call
\begin{gather*}
\tau\bigl(2{\pb^{\rm B}},2\tilde{\pb}^{\rm B}|\hat{h}^\pm\bigr)=
\langle 0| \hat{\gamma}^{\rm B\pm}\bigl(2{\pb^{\rm B}}\bigr) \hat{h}^\pm
\hat{\gamma}^{\dag\rm B\pm}\bigl(2\tilde{\pb}^{\rm B}\bigr) |0\rangle
\end{gather*}
the two-sided BKP (2-BKP) tau function.
Here ${\pb^{\rm B}}=(p_{1},p_{3},p_{ 5},\dots)$
and $\tilde{\pb}^{\rm B}=(\tilde{p}_{1},\tilde{p}_{3},\tilde{p}_{ 5},\dots)$
are parameters. We call $t^{\rm B}_m=\frac 2m p^{\rm B}_m$ and
$\tilde{t}^{\rm B}_m=\frac 2m \tilde{p}^{\rm B}_m$ the two-sided BKP higher times.
\end{de}
This notion of the BKP higher times is the same as it was introduced in \cite{DJKM1,DJKM2} and \cite{KvdL1} and is common.

Two-sided BKP tau function is usual (one-sided) BKP tau function (\ref{tau-BKP}) with respect to the higher times
$\tb^{\rm B}=\bigl(t^{\rm B}_1,t^{\rm B}_3,\dots\bigr)$
and at the same time it is one-sided KP tau function with respect to the higher time variables
$\tilde{\tb}^{\rm B}=\bigl(\tilde{t}^{\rm B}_1,\tilde{t}^{\rm B}_3,\dots\bigr)$, see Remark \ref{adjointB}.

\begin{rem}
The case $\tilde{N}=0$, see (\ref{p=p(x,y)_tilde}) and (\ref{p_B_tilde})
is related to the (one-side) tau function.
\end{rem}

\begin{de}
We call
\be
K_\mu\bigl({\pb^{\rm B}} |\hat{h}^\pm\bigr):=
\langle 0 |\hat{\gamma}^{{\rm B}\pm}\bigl({\pb^{\rm B}}\bigr) \hat{h}^\pm \Phi_\mu| 0 \rangle,
\label{K}
\ee
 lattice BKP tau function labeled with a partition $\mu\in\DP$.
\end{de}

We do not assume that $K_\lambda\bigl(\pb^{\rm B}\bigr)$ is a polynomial function in the variables
$\pb^{\rm B}$. For the polynomial case, see \cite{HO2}.

With the assumption of Remark \ref{h_entries_exist}, one can write both BKP tau functions as Taylor series
in power sum variables as follows:
\begin{gather}
\tau\bigl(2{\pb^{\rm B}},2\tilde{\pb}^{\rm B}|\hat{h^\pm}\bigr)
 =\sum_{\mu,\nu\in\DP} \hat{h}_{\mu,\nu}^\pm Q_\mu\bigl({\pb^{\rm B}}\bigr)Q_\nu\bigl(\tilde{\pb}^{\rm B}\bigr), \nonumber
\\ \label{K_Taylor}
 K_\mu\bigl({\pb^{\rm B}} |\hat{h}^\pm\bigr)=\sum_{\nu\in\DP} \hat{h}_{\mu,\nu}^\pm Q_\nu\bigl({\pb^{\rm B}}\bigr),
\end{gather}
which results from (\ref{coherent_states}). For a given $\hat{h}^\pm$, we obtain
\[
\tau\bigl(2{\pb^{\rm B}},2\tilde{\pb}^{\rm B}|\hat{h}^\pm\bigr)
 =\sum_{\mu\in\DP} 2^{-\ell(\mu)} K_\mu\bigl({\pb^{\rm B}} |\hat{h}^\pm\bigr)Q_\mu\bigl(\tilde{\pb}^{\rm B}\bigr).
\]
For $\hat{h}^\pm=1$ (the same, $B=0$), we get the simplest two-sided BKP tau function~(\ref{vacuum-2-BKP'}).

\begin{rem}\label{Cartan}
In the Grassmannian approach to the BKP tau functions developed in \cite{HB}, the
coefficients $K_\lambda\bigl({\pb^{\rm B}}|\hat{h}^\pm\bigr)$
in (\ref{2-KP_lattice-KP}) have the meaning of the Cartan coordinates $\kappa_\lambda\bigl(\hat{h}^\pm\bigr)$.
The notation was used in \cite{HO2} instead of \smash{$K_\lambda\bigl(\pb|\hat{h}^\pm\bigr)$}, see Remark \ref{notational_remark}.
\end{rem}

{\bf The known relation between KP and BKP two-sided tau functions.}
It is well known~\cite{JM} that the square of a (one-side) BKP tau function is equal to a certain
(one-side) KP tau function under
restriction $p_{2m}\equiv 0$;
the same is obviously true for two-sided tau functions.
\begin{Theorem}
\be\label{2-BKP2-BKP=2-KP}
\tau^{\rm B}\bigl(2{\pb^{\rm B}},2\tilde{\pb}^{\rm B}|\hat{h}^+\bigr)\tau^{\rm B} \bigl(2{\pb^{\rm B}},2\tilde{\pb}^{\rm B}|\hat{h}^-\bigr)=
\tau\bigl(\pb',\tilde{\pb}'|\hat{h}^+\hat{h}^-\bigr),
\ee
where $\tau^{\rm B}\bigl(2{\pb^{\rm B}},2\tilde{\pb}^{\rm B}|\hat{h}^+\bigr)=
\tau^{\rm B} \bigl(2{\pb^{\rm B}},2\tilde{\pb}^{\rm B}|\hat{h}^-\bigr)$, ${\pb^{\rm B}}$ is given in
\eqref{p_B},
\begin{gather*}
\tilde{\pb}^{\rm B}=\bigl(\tilde{p}^{\rm B}_1,\tilde{p}^{\rm B}_3,\tilde{p}^{\rm B}_5,\dots\bigr)
\end{gather*}
\big(with $2\tilde{p}^{\rm B}_i=\tilde{p}_i$, $i$ odd\big)
are two given sets of parameters and
\begin{gather}\label{p=p'}
\pb':=(p_1,0,p_3,0,p_5,0,\dots),
\\ \label{tilde-p=p'}
\tilde{\pb}':=(\tilde{p}_1,0,\tilde{p}_3,0,\tilde{p}_5,0,\dots) .
\end{gather}
\end{Theorem}

The proof is identical to the similar statement \cite{JM} about KP and BKP tau functions with one set of variables
and based on Lemma \ref{factorization_lemma} and on the relation
\[
\hat{\gamma}(\pb')=\hat{\gamma}^{{\rm B}+}\bigl(2{\pb^{\rm B}}\bigr)\hat{\gamma}^{{\rm B}-}\bigl(2{\pb^{\rm B}}\bigr),\qquad
\hat{\gamma}^\dag(\pb')=\hat{\gamma}^{\dag{\rm B+}}\bigl(2{\pb^{\rm B}}\bigr)\hat{\gamma}^{\dag{\rm B-}}\bigl(2{\pb^{\rm B}}\bigr)
\]
(which is the direct result of (\ref{J=J^B+J^B}), (\ref{Hei-B}) and can be treated as an example of
Lemma \ref{g=h^+h^-Lemma}).

Our goal will be to get the generalization of (\ref{2-BKP2-BKP=2-KP}) in Section \ref{Results},
where now $\pb=\pb(\ab/{-}\bb)$ and $\tilde{\pb}=\tilde\pb\bigl(\tilde\ab/{-}\tilde\bb\bigr)$
of (\ref{pb(xb/-yb)}) and (\ref{tilde-pb(xb/-yb)}). This choice of power sums is different from
(\ref{p=p'}) and~(\ref{tilde-p=p'}) except the case mentioned in Remark \ref{when-p=p'}.

{\bf The known relation between KP and BKP lattice tau functions.}
In the previous work \cite{HO2}, we have shown that in case the power sums are restricted according to (\ref{p=p'})
we get the following theorem.

\begin{Theorem}[\cite{HO2}]\label{ThHO2} For ${\pb^{\rm B}}$ and $\pb$ given by \eqref{p_B}, \eqref{p=p'} and \eqref{p_p^B}, respectively, we have
\[
 S_{(\alpha|\beta)}\bigl(\pb|\hat{h}^+\hat{h}^-\bigr) =
\sum_{ \zeta\in\PP(\alpha,\beta)}
\left[\zeta^+,\zeta^-\atop \alpha,\beta \right]
K_{\zeta^+} \bigl(2{\pb^{\rm B}}|\hat{h}^+\bigr)
K_{\zeta^-} \bigl(2{\pb^{\rm B}}|\hat{h}^-\bigr),
\]
where the factor in square brackets is given by \eqref{a}.
In particular,
\[
 s_{(\alpha|\beta)}(\pb') =
\sum_{ \zeta\in\PP(\alpha,\beta)}
 \left[\zeta^+,\zeta^-\atop \alpha,\beta \right]
Q_{\zeta^+}\bigl({\pb^{\rm B}}\bigr)
Q_{\zeta^-}\bigl({\pb^{\rm B}}\bigr)
\]
obtained in {\rm \cite{HO1}}.
\end{Theorem}

Proof follows from Lemmas \ref{Psi-lambda-Phi-mu-Phi-mu} and~\ref{factorization_lemma}.

In \cite{HO2}, we consider KP polynomial tau functions as the bilinear sum of BKP polynomial tau functions
\cite{KvdLRoz}; this relation is given by Theorem~\ref{ThHO2}.

In Section \ref{final-bilinear-relations},
we are going to replace the restriction $\pb=\pb'$ by $\pb=\pb(\ab/{-}\bb)$ and present the generalization
of Theorem \ref{ThHO2}.

\subsection{On lattice polynomial KP and lattice polynomial BKP tau functions}

Among lattice tau functions, there is a family of tau functions polynomial in power sum variables.
Polynomial tau functions were studied in \cite{KvdLRoz} and also in \cite{HO3}.

To be precise let us introduce the following definition.

\begin{de}
We call a lattice KP tau function (\ref{S_Taylor}) polynomial if for {\em any} given $\lambda \in\Pa$
the number of matrix elements
$ \hat{g}_{\mu,\lambda}$ is finite.
\end{de}

\begin{de}
We call a lattice BKP tau function (\ref{K_Taylor}) polynomial if for {\em any} given $\alpha \in\DP$
the number of matrix elements
$ \hat{h}_{\nu,\alpha}$ is finite.
\end{de}

\begin{rem}
In the Grassmannian approach, the property of polynomiality may be formulated as follows:
It is known that for any $\lambda\in\Pa$ there exits such $\hat{g}(\lambda)\in {\rm GL}_\infty$ that
$\hat{g}(\lambda)|0\rangle = \Psi_\lambda|0\rangle$ (see Remark~\ref{g,h,Psi,Phi}),
and perhaps it will be better to attach a lattice KP tau function (\ref{largeS}) to the Schubert cell
labeled with the partition $\lambda$. Then one can write
$\hat{g}=\hat{g}_L\hat{g}_0\hat{g}_R\hat{g}(\lambda)$,
where $\hat{g}_L,\hat{g}_0,\hat{g}_R \in {\rm GL}_\infty$,
$\hat{g}_R|0\rangle=|0\rangle$,
$\langle 0|\hat{g}_L=\langle 0|$,
$\hat{g}_0|0\rangle =g_0|0\rangle$, $\langle 0|\hat{g}_0 = g_0\langle 0|$,
$g_0$ is a number.
Then the condition of the polynomiality is the nilpotent structure of $\hat{A}_L$, where
\smash{$\hat{g}_L={\rm e}^{\hat{A}_L}$}.

The similar condition of the polynomiality for (\ref{K}) we get for
$\hat{h}^\pm=\hat{h}^\pm_L\hat{h}^\pm_0\hat{h}^\pm_R\hat{h}^\pm(\alpha)$,
where \smash{$\hat{h}^\pm(\alpha)|0\rangle = \Phi^\pm_{\hat\alpha}|0\rangle $} (see \cite{HB}),
$\hat{h}^\pm_R|0\rangle=|0\rangle$,
$\langle 0|\hat{h}^\pm_L=\langle 0|$, $\hat{h}^\pm_0|0\rangle =h_0|0\rangle$, $\langle 0|\hat{h}^\pm_0 =h_0\langle 0|$,
 $h_0$ is a number.
Then the condition of the polynomiality is the nilpotent structure of $\hat{B}^\pm_L$, where
$\hat{h}^\pm_L={\rm e}^{\hat{B}^\pm_L}$.

Let us note that in \cite{HO3} we actually consider only the cases where $\hat{g}_L=1$ and
$\hat{h}^\pm_L=1$.
\end{rem}

\section{Results} \label{Results}

\subsection{Key lemmas} \label{Keb_lemmas}

Let us note that tau functions below depend on $\tilde{\Psi}(\ab/{-}\bb)$ and
$\tilde{\Phi}(\ab)$. And as one can notice,
 $\tilde{\Psi}(\ab/{-}\bb)$ and $\tilde{\Phi}(\ab)$ do not depend on the order of the arguments
 $a_1,\dots,a_N,b_1,\dots,b_N$. Nevertheless, as we have already wrote, it is suitable to order
 these sets to have a unified approach for problems under consideration.

\begin{lem}\label{Polarization-fields}
\begin{gather}\label{Psi=PhiPhi}
\Psi(\ab/{-}\bb)=
\sum_{(\zb^+,\zb^-)\in \PP(\ab, \bb)} \textsc{d}^{\zb^+,\zb^-}_{\ab,\bb} \Phi^+(\zb^+)\Phi^-(\zb^-) ,
\\ \label{tilde-Psi=PhiPhi}
\tilde{\Psi}(\ab/{-}\bb)=
\sum_{(\zb^+,\zb^-)\in \PP(\ab, \bb)} \tilde{\textsc{d}}^{\zb^+,\zb^-}_{\ab,\bb}
\tilde{\Phi}^+(\zb^+)\tilde{\Phi}^-(\zb^-) ,
\end{gather}
where
\begin{gather*}
 \textsc{d}^{\zb^+,\zb^-}_{\ab,\bb}:={(-1)^{\frac 12 N(N+1)+\tilde{q}}\over 2^{N-\tilde{q}}}
\sgn(\zb)(-1)^{\pi(\zb^-)} {\rm i}^{ N(\zb^-)},
\\
\tilde{\textsc{d}}^{\zb^+,\zb^-}_{\ab,\bb}=\textsc{d}^{\zb^+,\zb^-}_{\ab,\bb}
\frac{\Delta^{\rm B}(\zb^+)\Delta^{\rm B}(\zb^-)}{\Delta(\ab,\bb)} .
\end{gather*}
\end{lem}

\begin{proof}
 Consider equation~(\ref{Psi-x1}) and substitute
 \begin{gather*}
\psi(a_j) = \frac{1} {\sqrt{2}}\bigl( \phi^+(a_j) - {\rm i} \phi^-(a_j)\bigr), \\
 \psi^\dag(-b_j)= \frac{1}{ b_j\sqrt{2}} \bigl(\phi^+(b_j) + {\rm i} \phi^-(b_j)\bigr), \qquad j=1,\dots,N,
\end{gather*}
see (\ref{psi(z)-phi(z)}),
for all the factors, and expand the product as a sum over monomial terms of the form
\[
 \prod_{j=1}^{N(\zb^+)} \phi^+\bigl(z^+_j\bigr) \prod_{k=1}^{N(\zb^-)} \phi^-(z^-_k) .
\]
We have in mind formula (\ref{psipis=phiphi-fields}) if $b_k=a_j$
which reduces the number of terms in the product of Fermi fields.

Taking into account the sign factor $\sgn(\zb)$ corresponding to the order of the neutral fermion
factors, as well as the powers of $-1$ and ${\rm i}$, and noting that there are $2^{\tilde{q}}$ resulting identical terms,
then gives (\ref{Psi=PhiPhi}), then (\ref{tilde-Psi=PhiPhi})
\end{proof}

\begin{lem}\label{Polarization-fields*}
\begin{gather*}\label{Psi=PhiPhi*}
\Psi(\bb^*/{-}\ab^*)=
\sum_{(\zb^+,\zb^-)\in \PP(\bb^*, \ab^*)} \textsc{d}^{\zb^+,\zb^-}_{\bb^*,\ab^*} \Phi^{+}(\zb^+)\Phi^{-}(\zb^-) ,
\\ \label{tilde-Psi=PhiPhi*}
\tilde{\Psi}(\bb^*/{-}\ab^*)=
\sum_{(\zb^+,\zb^-)\in \PP(\bb^*, \ab^*)} \tilde{\textsc{d}}^{\zb^+,\zb^-}_{\bb^*,\ab^*}
\tilde{\Phi}^+(\zb^+)\tilde{\Phi}^-(\zb^-) ,
\end{gather*}
where
\begin{gather*}
 \textsc{d}_{\bb^*,\ab^*}^{\zb^+,\zb^-} :={(-1)^{\frac 12 N(N+1)+q}\over 2^{N-q}}
\sgn(\zb)(-1)^{\pi(\zb^+)} {\rm i}^{ N(\zb^+)},
\\
\tilde{\textsc{d}}^{\zb^+,\zb^-}_{\bb^*,\ab^*}=\textsc{d}^{\zb^+,\zb^-}_{\bb^*,\ab^*}
\frac{\Delta^{\rm B}(\zb^+)\Delta^{\rm B}(\zb^-)}{\Delta(\bb^*/{-}\ab^*)} .
\end{gather*}
\end{lem}

\begin{proof}
 Consider equation~(\ref{Psi-x1}) and substitute
 \begin{gather}
\psi\bigl(-a_j^{-1}\bigr) = \frac{1} {\sqrt{2}}\bigl( \phi^+\bigl(-a_j^{-1}\bigr) - {\rm i} \phi^-\bigl(-a_j^{-1}\bigr)\bigr), \nonumber\\
 \psi^\dag\bigl(b_j^{-1}\bigr)= \frac{b_j}{ \sqrt{2}} \bigl(\phi^+\bigl(-b_j^{-1}\bigr) + {\rm i} \phi^-\bigl(-b_j^{-1}\bigr)\bigr), \qquad j=1,\dots,N,
\end{gather}
see (\ref{psi(z)-phi(z)}),
for all the factors, and expand the product as a sum over monomial terms of the form
\[
 \prod_{j=1}^{N(\zb^+)} \phi^+\bigl(z^+_j\bigr) \prod_{k=1}^{N(\zb^-)} \phi^-(z^-_k) .
 \]

We have in mind formula (\ref{psipis=phiphi-fields'}) if ${b}_k={a}_j$,
which reduces the number of terms in the product of Fermi fields.
 Taking into account the sign factor $\sgn(\zb)$ corresponding to the order of the neutral fermion
factors, as well as the powers of $-1$ and ${\rm i}$, and noting that there are $2^{q}$ resulting identical terms,
then gives (\ref{Psi=PhiPhi*}). Then (\ref{tilde-Psi=PhiPhi*}) follows.
\end{proof}

 For $\lambda=(\alpha|\beta)$, we have the following lemma.

 \begin{lem}\label{Psi-lambda-Phi-mu-Phi-mu}
\be
\Psi_\lambda = {(-1)^{\frac{1}{2}r(r+1) + s}\over 2^{r-s}} \sum_{\zeta=(\zeta^+,\zeta^-)\in \PP(\alpha, \beta)}
\sgn(\zeta)(-1)^{\pi(\zeta^-)} {\rm i}^{ m(\zeta^-)}
\Phi_{\zeta^+} \Phi_{\zeta^-} .
\label{polariz_sum}
\ee
\end{lem}
\begin{proof}
 In equation~(\ref{Psi^dag-lambda}), reorder the product over the factors $\psi_{\alpha_j} \psi^\dag_{-\beta_j-1}$ so
 the $\psi_{\alpha_j}$ terms precede the \smash{$\psi^\dag_{-\beta_j-1}$} ones, giving an overall sign
 factor \smash{$(-1)^{\frac{1}{2}r(r-1)}$}. Then substitute
 \[
\psi_{\alpha_j} = \frac{1} {\sqrt{2}}\bigl( \phi^+_{\alpha_j} - {\rm i} \phi^-_{\alpha_j}\bigr), \qquad
 \psi^\dag_{-\beta_j-1}= \frac{(-1)^j}{ \sqrt{2}} \bigl(\phi^+_{\beta_j+1} + {\rm i} \phi^-_{\beta_j+1}\bigr), \qquad j\in \Zb,
\]
which follows from (\ref{charged-neutral}),
for all the factors, and expand the product as a sum over monomial terms of the form
\[
 \prod_{j=1}^{m(\zeta^+)} \phi^+_{\zeta^+_j} \prod_{k=1}^{m(\zeta^-)} \phi^-_{\zeta^-_k} .
 \]
 Taking into account the sign factor $\sgn(\zeta)$ corresponding to the order of the neutral fermion
factors, as well as the powers of $-1$ and $i$, and noting that there are $2^s$ resulting identical terms,
then gives (\ref{polariz_sum}).
\end{proof}

For $\lambda=(\alpha|\beta)$, we also have the following lemma.

 \begin{lem}
\begin{gather*}
\Psi^\dag_\lambda = {(-1)^{\frac{1}{2}r(r+1) + s}\over 2^{r-s}}(-1)^{|\lambda|}
\sum_{\zeta=(\zeta^+,\zeta^-)\in \PP(\alpha, \beta)}
\sgn(\zeta)(-1)^{\pi(\zeta^-)} {\rm i}^{ m(\zeta^-)}
\Phi_{-\zeta^+} \Phi_{-\zeta^-} .
\end{gather*}
\end{lem}

\section{Bilinear expressions}

Below, we assume the condition
$
\hat{g}=\hat{h}^+\hat{h}^-$, and we use Lemma \ref{factorization_lemma} and results of Section \ref{Keb_lemmas}.

\subsection{KP tau functions as a bilinear expression in BKP tau functions}\label{final-bilinear-relations}

Let
\begin{gather*}
p_m(\ab/{-}\bb)=
\sum_{j=1}^N \bigl( a_j^m - (-b_j)^m \bigr),\qquad m>0, \\
\tilde{p}_{m}(\tilde{\ab}/{-}\tilde{\bb}) =
\sum_{j=1}^{\tilde{N}}\bigl( \tilde{a}_j^m - (-\tilde{b}_j)^m \bigr)
,\qquad m>0.
\end{gather*}
For a given $\zb=\bigl(z_1,\dots,z_{N(\zb)}\bigr)$, we define
\begin{gather*}
p^{\rm B}_m(\zb)=
\sum_{j=1}^{N(\zb)} z_j^m ,\qquad
\tilde{p}^{\rm B}_{m}(\tilde{\zb})=
\sum_{j=1}^{m(\tilde{\zb})} \tilde{z}_j^m
,\qquad m>0.
\end{gather*}

{\bf Two-sided KP tau function as a bilinear expression in two-sided BKP tau functions.}
Then
\begin{gather*}
\tau\bigl(\pb(\ab/{-}\bb),\tilde{\pb}\bigl(\tilde{\ab}/{-}\tilde{\bb}\bigr)|\hat{g}\bigr)=
\langle 0| \tilde{\Psi}(\bb^*/{-}\ab^*) \hat{g} \tilde{\Psi}\bigl(\tilde{\ab}/{-}\tilde{\bb}\bigr)|0\rangle
=:\tau\bigl(\ab/{-}\bb;\tilde{\ab}/{-}\tilde{\bb}|\hat{g}\bigr),
\\
\tau^{{\rm B}\pm}\bigl({\pb^{\rm B}}(\zb),\tilde{\pb}^{\rm B}(\tilde{\zb})|\hat{h}^\pm\bigr):=
\langle 0| \tilde{\Phi}^{\pm}(-\zb^*) \hat{h}^{\pm} \tilde{\Phi}^\pm(\tilde{\zb})|0\rangle =:
\tau^{{\rm B}\pm}\bigl(\zb;\tilde{\zb}|\hat{h}^\pm\bigr) .
\end{gather*}

\begin{Theorem}
Suppose $\hat{g}=\hat{h}^{+}\hat{h}^{ -}$. Then
\[
\tau\bigl(\ab/{-}\bb;\tilde{\ab}/{-}\tilde{\bb}|\hat{h}^{+}\hat{h}^{ -}\bigr)
=\sum_{{\zb}\in\PP({\ab},{\bb}) \atop \tilde{\zb}\in\PP(\tilde{\ab},\tilde{\bb})}
\left[\zb^-,\zb^+\atop \ab,\bb\right]\left[\tilde{\zb}^-,\tilde{\zb}^+\atop \tilde{\ab},\tilde{\bb}\right]^*
\tau^{{\rm B}\pm}\bigl(\zb^+;\tilde{\zb}^+|\hat{h}^+\bigr)\tau^{{\rm B}\pm}\bigl(\zb^-;\tilde{\zb}^-|\hat{h}^-\bigr) ,
\]
where the factors in square brackets are given by \eqref{d}.
\end{Theorem}

Proof follows from Lemmas \ref{Polarization-fields}, \ref{Polarization-fields*}
and \ref{factorization_lemma}.

Let us consider the case $\hat{g}=1$ (i.e., $\hat{g}_{\mu,\lambda}=\delta_{\mu,\lambda}$ for
$\mu,\lambda\in\Pa$) and $\bb=\ab$, \smash{${\tilde{\ab}=\tilde{\bb}}$,~${N=1}$}, then
\[
 \tau\bigl(\ab/{-}\bb;\tilde{\ab}/{-}\tilde{\bb}\bigr) =\frac{1}{(x+y)(\tilde{x}+\tilde{y})},
\]
$z^\pm_1=x$, $ z^\mp_2=y$, $ \tilde{z}^\pm_1=\tilde{x}$, $ \tilde{z}^\mp_2=\tilde{y}$.

{\bf Lattice KP tau function as a bilinear expression in the lattice BKP tau functions.}
Denote
$\pi_{(\alpha|\beta)}(\hat{g}) (\ab/{-}\bb):=
\langle 0 | \tilde{\Psi}(\bb^*/{-}\ab^*) \hat{g} \Psi_\lambda | 0 \rangle$,
labeled by partitions $\lambda =(\alpha|\beta)$ and
a $B_\infty$ lattice of BKP $\tau$-functions
\begin{gather*}
\kappa_{\alpha} (\hat{h})(\zb) :=
\langle 0 | \tilde{\Phi}^\pm(-\zb^*) \hat{h}^\pm \Phi^\pm_{\alpha} |0 \rangle .
\end{gather*}

\begin{rem}
The case $\lambda=0$ describes the usual (one-side) tau function.
\end{rem}

{\samepage

\begin{Theorem}\label{lattice-tau}
Suppose $\hat{g}=\hat{h}^{+}\hat{h}^{ -}$. Then
\[
 \pi_{(\alpha|\beta)}(\hat{g})(\ab/{-}\bb) =
\sum_{\zb\in\PP(\bb^*,\ab^*)\atop \mu\in\PP(\alpha,\beta)}
 \left[\zb^+,\zb^-\atop \ab,\bb\right] \left[ \zeta^+,\zeta^-\atop \alpha,\beta\right]
\kappa_{\zeta^+} \bigl(\hat{h}\bigr)(\zb^+)
\kappa_{\zeta^-} \bigl(\hat{h}\bigr)(\zb^-),
\]
where factors in square brackets are given by \eqref{a} and \eqref{d}.
In particular $(\hat{g}=1)$,
\[
 s_{(\alpha|\beta)}(\ab/{-}\bb) =
\sum_{\zb\in\PP(\bb^*,\ab^*)\atop \mu\in\PP(\alpha,\beta)}
\left[\zb^+,\zb^-\atop \ab,\bb\right] \left[ \zeta^+,\zeta^-\atop \alpha,\beta\right]
Q_{\zeta^+}(\zb^+)
Q_{\zeta^-}(\zb^-) .
\]
\end{Theorem}

Proof follows from Lemmas \ref{Psi-lambda-Phi-mu-Phi-mu}, \ref{Polarization-fields*} and \ref{factorization_lemma}.

}

{\bf Bi-lattice KP and bi-lattice BKP tau functions.}
In view of the notion of the lattice tau functions (\ref{lattice_KP_series}) and (\ref{lattice_BKP_series}),
it is natural to call $\hat{g}_{\lambda,\tilde{\lambda}}$ bi-lattice KP tau function and to call
$\hat{h}_{\mu,\tilde{\mu}}$ bi-lattice BKP tau function. Then we get the following theorem.
\begin{Theorem}
\[
\hat{g}_{(\alpha|\beta)),(\tilde{\alpha}|\tilde{\beta})}=
\sum_{(\zeta^+,\zeta^-)\in \PP(\alpha,\beta)\atop (\tilde{\zeta}^+,\tilde{\zeta}^-)\in \PP(\tilde{\alpha},\tilde{\beta})}
\left[\tilde{\zeta}^+,\tilde{\zeta}^-\atop \tilde{\alpha}, \tilde{\beta} \right]
\left[ \zeta^+,\zeta^-\atop \alpha,\beta \right]
 \hat{h}^-_{\zeta^-,\tilde{\zeta}^-}\hat{h}^+_{\zeta^+,\tilde{\zeta}^+}.
\]
\end{Theorem}

\begin{cor} For $\hat{g}=1$, we have
\[
\delta_{\alpha,\tilde{\alpha}}\delta_{\beta\tilde{\beta}}=
\sum_{(\zeta^+,\zeta^-)\in \PP(\alpha,\beta)\atop (\tilde{\zeta}^+,\tilde{\zeta}^-)\in \PP(\tilde{\alpha},\tilde{\beta})}
\left[\tilde{\zeta}^+,\tilde{\zeta}^-\atop \tilde{\alpha}, \tilde{\beta} \right]
\left[ \zeta^+,\zeta^-\atop \alpha,\beta \right]
 \delta_{\zeta^-,\tilde{\zeta}^-}\delta_{\zeta^+,\tilde{\zeta}^+}.
\]
\end{cor}

The simple nontrivial example of the bi-lattice KP tau function $\hat{g}_{\lambda,\tilde{\lambda}}$ is the product
$s^*_\mu(\lambda)\dim\,\lambda$ where $s^*_\mu(\lambda)$ is the so-called shifted Schur function
introduced by Okounkov in \cite{EOP} and $\dim\lambda$ is the number of standard tableaux of the shape
$\lambda$, see \cite{Mac1}. The shifted Schur functions were used in an approach to the representation theory
developed by G. Olshanski and A. Okounkov in~\cite{OO}.
The simple nontrivial example of the bi-lattice BKP tau function $h_{\mu,\nu}$
is the product $Q^*_\mu(\nu) \dim^{\rm B}\mu$, where $Q^*_\mu(\nu)$ is the shifted projective Schur function
introduced by Ivanov in \cite{Iv} and $\dim^{\rm B}\mu$ is the number of the shifted standard tableaux of
the shape $\mu$. Functions $Q^*_\mu(\nu)$ are of use in the study of spin Hurwitz numbers \cite{MMN2019}.

The relation between shifted Schur functions and the shifted projective Schur functions was written
done in \cite{O-2021}.

\section[Summary of other formulas {[30]}]{Summary of other formulas \cite{O-2021}}

\subsection[Relation between characters of symmetric group and characters of Sergeev group]{Relation between characters of symmetric group \\ and characters of Sergeev group}

For a given $\Delta\in\OP$, one can split its parts into two ordered odd partitions $(\Delta^+,\Delta^-)$:
$\Delta=\Delta^+\cup\Delta^-$, $\Delta^+,\Delta^- \in OP$,
$\ell(\Delta^+) + \ell(\Delta^-)=\ell(\Delta)$. The set of all such $(\Delta^+,\Delta^-)$ we denote
by $O\PP(\Delta)$.

From $J_n=J^{\rm B+}_n+J^{\rm B^-}_n,\,n\,\text{odd}$ (see \cite{JM}), we obtain the following lemma.
\begin{lem} \label{Polarization-currents}
\[
J_\Delta = \sum_{(\Delta^+,\Delta^-)\in O\PP} J^{{\rm B}+}_{\Delta^+}J^{{\rm B}-}_{\Delta^-},
\qquad
J_{-\Delta} = \sum_{\Delta^+\in \OP\atop \Delta^+\cup \Delta^-=\Delta}
J^{{\rm B}+}_{-\Delta^+}J^{{\rm B}-}_{-\Delta^-},
\]
\end{lem}

It is known \cite{Mac1} that the power sums labeled by partitions
$\pb_\Delta =p_{\Delta_1}p_{\Delta_2}\cdots$, $\Delta\in\Pa$
(here $ \pb = (p_1, p_2, p_3, \dots) $) are uniquely expressed
in terms of Schur polynomials
\be\label{p-s-chi}
\pb_\Delta =\sum_{\lambda\in\Pa} \chi_\lambda(\Delta) s_\lambda(\pb),
\ee
while the odd power sum variables (power sums labeled by odd numbers)
$\pb_\Delta =p_{\Delta_1}p_{\Delta_2}\cdots$, $\Delta\in\OP$
also denoted by $\pb^{\rm B}_\Delta$ \big(where ${\pb^{\rm B}}=(p_1,p_3,p_5,\dots)$\big)
 are uniquely expressed in terms of projective Schur polynomials
\begin{gather*}
\pb^{\rm B}_\Delta =\sum_{\alpha\in\DP} \chi^{\rm B}_\alpha(\Delta) Q_\alpha\bigl({\pb^{\rm B}}\bigr)=\pb_\Delta,
\qquad \Delta \in \OP.
\end{gather*}

Let me recall that the coefficients $\chi_\lambda(\Delta)$ in (\ref{p-s-chi}) has the meaning of
the irreducible characters of the symmetric group $S_d$, $d=|\lambda|$ evaluated on the cycle class $C_\Delta$,
$\Delta=(\Delta_1,\dots,\Delta_k) $, ${|\lambda|=|\Delta|=d}$, see \cite{Mac1}, and we can write it as
$\chi_\lambda(\Delta) = \langle 0| J_\Delta \Psi_\lambda |0\rangle$,
where $J_\Delta $ and $\Psi_\lambda$ are given by~(\ref{J-Delta}) and~(\ref{Psi-lambda}), respectively
(see, for instance, \cite{MMNO} for details).

The characters of symmetric groups have a very wide application,
in particular, in mathematical physics. I will give two works as an example \cite{ItMirMor,MMS-knots}.

The notion of Sergeev group was introduced in \cite{EOP}.
The coefficient $\chi^{\rm B}_\alpha$ is the irreducible char\-ac\-ter of this group \cite{EOP,Serg}.
As it was shown in \cite{EOP}, the so-called spin Hurwitz numbers (introduced in this work) are expressed in terms
of these characters.
As it was pointed out in~\cite{Lee2018,MMN2019}, the generating function for spin Hurwitz numbers
can be related to the BKP hierarchy in a similar way as the generating function for usual Hurwitz numbers
is related to the KP (and also to Toda lattice) hierarchies \cite{MMN2011, Okoun2000,Okounkov-Pand-2006}.

We can write these characters in terms of the BKP currents $J^{\rm B}_m$ ($m$ odd)
\[
 \chi^{\rm B}_\alpha(\Delta) =2^{-\ell(\alpha)} \langle 0|J^{\rm B}_{\Delta}\Phi_{\alpha} |0\rangle,
\]
see \cite{MMNO}.

From Lemma \ref{Psi-lambda-Phi-mu-Phi-mu} and (\ref{Polarization-currents}), we obtain the following theorem.

\begin{Theorem}
The character $\chi_\lambda,\,\lambda=(\alpha|\beta)$ evaluated on an odd cycle $\Delta\in\OP$ is the bilinear
function of the Sergeev characters as follows:
\[
\chi_{(\alpha|\beta)}(\Delta)=
\sum_{ (\nu^+,\nu^-)\in\PP(\alpha,\beta) \atop (\Delta^+,\Delta^-)\in O\PP(\Delta)} 2^{\ell(\nu^+) + \ell(\nu^+) }
\left[ \nu^+,\nu^- \atop \alpha,\beta\right]
\chi^{\rm B+}_{\nu^+}\bigl(\Delta^+\bigr)\chi^{\rm B-}_{\nu^-}(\Delta^-) .
\]
\end{Theorem}

\subsection[Relation between generalized skew Schur polynomials and generalized projective skew Schur polynomials]{Relation between generalized skew Schur polynomials \\ and generalized projective skew Schur polynomials}

This section may be treated as a remark to \cite{HO1}.
Let us find the relation between the following quantities:
\begin{gather}\label{skew-s}
s_{\lambda/\mu}(\pb'):=\langle 0|\Psi^*_\mu\hat{\gamma}(\pb') \Psi_\lambda |0\rangle,\\
\label{skew-Q}
Q_{\nu/\theta}\bigl({\pb^{\rm B}}\bigr):=\langle 0|\Phi_{-\theta}\hat{\gamma}^{\rm B}\bigl(2{\pb^{\rm B}}\bigr) \Phi_\nu|0\rangle ,
\end{gather}
where the power sum variables are as follows: $\pb'=(p_1,0,p_3,0,p_5,0,\dots)$, ${\pb^{\rm B}}=\frac12(p_1,p_3,p_5,\dots)$,
where $\lambda=(\alpha|\beta)$, $\mu=({\gamma}|\delta)$ and where $\alpha,\beta,{\gamma},\delta,\theta,\nu\in\DP$.

 \begin{Theorem}\label{skew}
 \[
 s_{\lambda/\mu}(\pb') = \sum_{(\nu^+,\nu^-)\in \PP(\alpha,\beta)\atop (\theta^+,\theta^-)\in \PP({\gamma},\delta)}
 \left[\nu^+,\nu^-\atop\alpha,\beta\right] \left[\theta^+,\theta^-\atop {\gamma},\delta\right] Q_{\nu^+/\theta^+}\bigl({\pb^{\rm B}}\bigr)Q_{\nu^-/\theta^-}\bigl({\pb^{\rm B}}\bigr).
 \]
 \end{Theorem}

 For the proof, we apply Lemma \ref{Psi-lambda-Phi-mu-Phi-mu} to the first equality (\ref{skew-s})
 and consider all possible parities of $m\bigl(\nu^\pm\bigr),m\bigl(\theta^\pm\bigr)$ to apply Lemma \ref{factorization_lemma}
 and take into account the fermionic expression (\ref{skew-Q}).

 Theorem \ref{skew} in a straightforward way can be generalized for the generalized skew Schur and
 skew projective Schur polynomials which we define as follows.
 Suppose
 \[
 \hat{g} =\hat{g}(C)={\rm e}^{\sum C_{ij}\psi_i \psi^\dag_j } ,
 \]
 where the entries $C_{ij}$ form a matrix.
 Also suppose
 \[
\hat h^\pm=\hat h^\pm(A)={\rm e}^{\sum A_{ij}\phi_i^\pm \phi_j^\pm} ,
 \]
 where now the entries $A_{ij}=-A_{ji}$ form a skew-symmetric matrix.
Suppose that
\be\label{Borel}
\hat g |0\rangle = |0\rangle c_1,
\qquad
\hat h^\pm |0\rangle = |0\rangle c_2 ,\qquad c_{1,2}\in \mathbb{C},\qquad c_{1,2}\neq 0.
\ee

 Define the generalized skew Schur polynomials and the generalized projective skew Schur polynomials (the same, generalized skew Schur's $Q$-functions) by
\begin{gather}\label{skew-s-A}
s_{\lambda/\mu}(\pb|\hat g):=\langle \mu|\hat{\gamma}(\pb)\hat g |\lambda\rangle,\\
\label{skew-Q-A}
Q_{\nu/\theta}\bigl({\pb^{\rm B}}|\hat h^\pm\bigr):=\langle 0|\Phi_{-\theta}\hat{\gamma}^{\rm B\pm}\bigl(2{\pb^{\rm B}}\bigr) \hat h^\pm \Phi_\nu|0\rangle.
\end{gather}
\begin{rem}
In \cite{HO2}, $s_{\lambda/\mu}(\pb|\hat g)$ was defined by $s_{\lambda/\mu}(\pb|C)$ and
$Q_{\nu/\theta}\bigl({\pb^{\rm B}}|\hat h^\pm\bigr)$ was defined by $Q_{\nu/\theta}\bigl({\pb^{\rm B}}|A\bigr)$.
\end{rem}
\begin{rem}
The constraints (\ref{Borel}) are sufficient, but not necessary, for the polynomiality condition of the right-hand sides in (\ref{skew-s-A}) and (\ref{skew-Q-A}).
\end{rem}

The polynomiality of (\ref{skew-s-A}) in $p_1,\dots,p_{|\lambda|-|\mu|}$ and the polynomiality
of (\ref{skew-Q-A}) in $p_1,\dots,p_{|\nu|-|\theta|}$
follows from (\ref{Borel}).
One can treat a given $\pb$ as a constant and study $s_{\lambda/\mu}(\pb|A)$ as the function of discrete
variables $\lambda_i-i$ and $\mu_i-i$.

 \begin{Theorem}
 Suppose $\hat g=\hat h^+ \hat h^-$ and \eqref{Borel} is true. Then
 \[
 s_{\lambda/\mu}(\pb'|\hat g) = \sum_{(\nu^+,\nu^-)\in \PP(\alpha,\beta)\atop (\theta^+,\theta^-)\in \PP({\gamma},\delta)}
 \left[\nu^+,\nu^-\atop\alpha,\beta\right] \left[\theta^+,\theta^-\atop {\gamma},\delta\right]
 Q_{\nu^+/\theta^+}\bigl({\pb^{\rm B}}|\hat h^+\bigr)Q_{\nu^-/\theta^-}\bigl({\pb^{\rm B}}|\hat h^-\bigr)
 \]
 \end{Theorem}
The proof is based on the same reasoning as the proof of the theorem in \cite{HO2}. We omit it.

For a given function $r$ of a single variable, introduce
\begin{gather*}
s_{\lambda/\mu}(\pb|\rb):=\langle 0|\Psi^*_\mu \hat{\gamma}_r(\pb) \Psi_\lambda|0\rangle,\\
Q_{\nu/\theta}({\pb^{\rm B}}|\rb):=\langle 0|\Phi_{-\theta}\hat{\gamma}_r^{\rm B\pm}(2{\pb^{\rm B}}) \Phi_\nu|0\rangle ,
\end{gather*}
where
\begin{gather*} 
\hat{\gamma}_r(\pb):={\rm e}^{\sum_{j=1}^\infty \frac 1j p_j \sum_{k\in\Zb} \psi_k \psi^\dag_{k+j} r(k+1)\cdots r(k+j)} ,
\\ 
\hat{\gamma}_r^{\rm B}(2{\pb^{\rm B}}):={\rm e}^{\sum_{j=1,\text{odd}}^\infty \frac 2j p_j
\sum_{k\in\Zb} (-1)^k\phi^\pm_{k-j} \phi^\pm_{-k} r(k+1)\cdots r(k+j)}.
\end{gather*}
By the Wick theorem, we obtain
\[
 s_{\lambda/\mu}(\pb|\rb)=\det \bigl( r(\mu_i-i+1) \cdots r(\lambda_j-j)
 s_{(\lambda_i-\mu_j-i+j)}(\pb) \bigr)_{i,j}.
\]
The similar relation for $Q_{\nu/\theta}\bigl({\pb^{\rm B}}|\rb\bigr)$ is more spacious, the Wick theorem
yields a Nimmo-type pfaffian formula, we
shall omit it. We have
\[
 s_{\lambda/\mu}(\pb'|\rb) = \sum_{(\nu^+,\nu^-)\in \PP(\alpha,\beta)\atop (\theta^+,\theta^-)\in \PP({\gamma},\delta)}
\left[\nu^+,\nu^-\atop \alpha,\beta\right] \left[\theta^+,\theta^-\atop {\gamma},\delta\right]
 Q_{\nu^+/\theta^+}\bigl({\pb^{\rm B}}|\rb\bigr)Q_{\nu^-/\theta^-}\bigl({\pb^{\rm B}}|\rb\bigr).
 \]

\subsection{Relation between shifted Schur and shifted projective Schur functions}

Let us recall the notion of the shifted Schur function introduced by Okounkov and Olshanski~\cite{OkOl}.
In can be defined as
\begin{gather*}
s^*_\mu(\lambda) = \frac{\dim\, \lambda/\mu}{\dim\,\lambda} n(n-1)\cdots (n-k+1)
=\frac{s_{\lambda/\mu}(\pb_1)}{s_{\lambda}(\pb_1)} ,
\end{gather*}
where $n=|\lambda|$, $k=|\mu|$, $\pb_1=(1,0,0,\dots)$ and
\[
\dim\, \lambda/\mu=s_{\lambda/\mu}(\pb_1)(n-k)!,\qquad \dim\, \lambda=s_{\lambda}(\pb_1)n!
\]
are the number of the standard tableaux of the shape $\lambda/\mu$ and $\lambda$ respectively, see \cite{Mac1}.
The function $s^*_\mu(\lambda)$ as a function of the Frobenius coordinates is also known as Frobenius--Schur
function $FS(\alpha,\beta)$

On the other hand, Ivanov \cite{Iv} introduced the projective analogue of shift $Q$-functions
\begin{gather*}
Q^*_\theta(\nu)
=\frac{Q_{\nu/\theta}(\pb_1)}{Q_{\nu}(\pb_1)}.
\end{gather*}
Therefore,
\[
s^*_\mu(\lambda)s_{\lambda}(\pb_1)=
\sum_{(\nu^+,\nu^-)\in \PP(\alpha,\beta)\atop (\theta^+,\theta^-)\in \PP({\gamma},\delta)}
 \left[ \nu^+,\nu^-\atop\alpha,\beta\right]
 \left[\theta^+,\theta^-\atop{\gamma},\delta\right]
 Q^*_{\theta^+}\bigl(\nu^+\bigr)Q^*_{\theta^-}(\nu^-)
 Q_{\nu^+}(\pb_1)Q_{\nu^-}(\pb_1).
\]

Let me add that both $s^*$ and $Q^*$ were used in the description of the generalized cut-and-join structure
\cite{MMN2019, MMN2011} in the topics of Hurwitz and spin Hurwitz numbers.

\appendix

\section{Wick's theorem}
\label{wick_app}

The following summarizes the implication of Wick's theorem for fermionic VEV's
in a form that is used repeatedly in this work (see, e.g., \cite[Section~5.1]{HB}).

 \begin{Theorem}[Wick's theorem] 
The vacuum expectation value of the product of an even number of linear elements $\{w_i\}_{1\le i \le 2n}$
 of the fermionic Clifford algebra is
\be
\langle 0| w_1w_2\cdots w_{2n}|0\rangle=
\Pf\left(W_{ij}\right)_{1\le i,j \le 2n} ,
\label{wick_pfaff}
\ee
where $\{W_{ij}\}_{1\le i,j \le 2n}\}$ are elements of the skew $2n \times 2n$ matrix defined by
\[
W_{ij} = \begin{cases}
\langle 0| w_i w_j|0\rangle & \text{if } i<j ,\\
 0 & \text{if } i=j ,\\
-\langle 0| w_i w_j|0\rangle & \text{if } i> j,
\end{cases}
\]
whereas the VEV of the product of an odd number vanishes
\[
\langle 0| w_1w_2\cdots w_{2n+1}|0\rangle = 0.
\]
 \end{Theorem}
 In particular, if half the $w_i$'s consist of creation operators $\{u_i\}_{i=1, \dots , n}$
 and the other half annihilation operators $\big\{v^\dag_i\big\}_{i=1, \dots , n}$, so that
 \[
 \langle 0| u_i u_j |0 \rangle =0, \qquad \langle 0| v^\dag_i v^\dag_j |0 \rangle =0,\qquad 1\le i,j \le n,
 \]
 then (\ref{wick_pfaff}) reduces to
 \begin{gather*}
 \langle 0| u_1 v^\dag_1 \cdots u_n v^\dag_n|0\rangle= \det \bigl(\langle 0 | u_i v^\dag_j |0 \rangle\bigr)_{1\le i, j \le n}.
 \label{wick_det}
 \end{gather*}

 \subsection*{Acknowledgements}

This work of A.O.\ is an output of a research project implemented as part of the Basic Research Program at the National Research University Higher School of Economics (HSE University). The author thanks the anonymous referees for their helpful comments.

\pdfbookmark[1]{References}{ref}
\LastPageEnding

\end{document}